\documentclass[11pt]{article}

\usepackage[utf8]{inputenc}
\usepackage{amsmath, amssymb, amsthm, mathtools}
\usepackage[authoryear]{natbib}
\usepackage{geometry}
\usepackage{setspace}
\usepackage{lmodern}
\usepackage{thmtools}
\usepackage{thm-restate}
\usepackage{hyperref}
\hypersetup{hidelinks} 

\renewcommand{\sigma}{J}

\renewcommand{\mu}{K}

\renewcommand{\lambda}{L}

\newcommand{\prob}{\mathbb{P}}
\newcommand{\Prob}{\mathbb{P}}
\newcommand{\diff}{\mathrm{d}}

\theoremstyle{plain}
\newtheorem{lemma}{Lemma}[section]
\newtheorem{proposition}{Proposition}[section]
\newtheorem{theorem}{Theorem}[section]

\theoremstyle{remark}
\newtheorem{remark}{Remark}[section]

\theoremstyle{definition}
\newtheorem{definition}{Definition}[section]
\newtheorem{example}{Example}[section]

\numberwithin{equation}{section}

\begin{document}

\title{Diversification and Stochastic Dominance:\\ When All Eggs Are Better Put in One Basket}
\author{Léonard Vincent\thanks{\texttt{Email address: leonard.vincent@protonmail.ch}}}
\maketitle


\begin{abstract}
\vspace{0.5em}
\noindent
Diversification is usually viewed as a reliable way to reduce risk, yet it can dramatically fail for heavy-tailed losses with infinite mean: pooling independent losses of this type may increase tail risk at every threshold. We study this reversal by comparing a diversified portfolio (a weighted average) of risks to a ``one-basket'' benchmark that concentrates the full exposure on a single component chosen at random according to the same weights. In the iid case, the benchmark reduces to a single risk, recovering the classical comparison between a single risk and a diversified portfolio. Our main result --- the \emph{one-basket theorem} --- provides new sufficient conditions under which the diversified portfolio has larger tail probabilities for all thresholds (first-order stochastic dominance) than this benchmark. The theorem enables weight-specific verification of the stochastic dominance relation and yields new applications, notably to averages of infinite-mean discrete Pareto risks. We further show that these failures of diversification are boundary cases of a general phenomenon: diversification always increases the likelihood of exceeding thresholds near zero, and under specific conditions this local effect extends to all thresholds, yielding first-order stochastic dominance.
\end{abstract}

\vspace{1em}
\noindent\textbf{Keywords:} diversification, risk pooling, stochastic order, infinite mean


\section{Introduction}
\label{sec:intro}
Over the past century, probabilistic modeling has become a cornerstone of risk management, offering a quantitative framework for assessing and mitigating risk. One of its significant contributions is its formal justification for diversification --- the practice of spreading exposure across multiple independent or weakly correlated risks to reduce overall variability. This justification rests on classical results such as the law of large numbers and Modern Portfolio Theory, which demonstrate that diversification can reduce volatility without sacrificing expected returns \citep{markowitz1952portfolio}. Taken together, these and related developments have come to be seen as confirming a longstanding belief in the benefits of diversification, a notion embedded in common sense and memorably captured by the proverb: ``Don't put all your eggs in one basket."

\subsection{Diversification as a source of risk}
Given both this theoretical and intuitive support, it may be surprising that in certain cases, diversification actually increases risk. A striking example of this phenomenon was recently established by \citet{chen2025unexpected}, who proved that diversification increases risk in the sense of first-order stochastic dominance when dealing with infinite-mean Pareto risks. More precisely, they showed that for $n \geq 2$ independent and identically distributed (iid) copies $X_1, \dots , X_n$ of a random variable $X$ following a Pareto distribution with shape parameter $\alpha \in (0,1]$, every weighted average of these variables is larger in first-order stochastic dominance than $X$ itself:
\begin{equation}
\label{eq:unexpected}
X \leq_{\text{st}} \theta_1 X_1 + \dots + \theta_n X_n,
\end{equation}
where $X \leq_{\text{st}} Y$ means $\prob(X > x) \leq \prob(Y > x)$ for all real $x$, i.e.\ $Y$ has larger exceedance (tail) probabilities than $X$ at every threshold.\\
\\
This result is particularly noteworthy for both the strength of the dominance relation it establishes and the relevance of the distribution it involves. On the one hand, first-order stochastic dominance is arguably the strongest form of stochastic comparison, with broad implications for decision-making. In particular, if the risks represent random losses, then relation~\eqref{eq:unexpected} implies that any decision-maker with a decreasing utility function over losses (i.e., who prefers smaller losses) would prefer holding the single risk $X$ rather than the diversified portfolio $\theta_1 X_1 + \dots + \theta_n X_n$ --- that is, they would choose to put all their eggs in one basket. On the other hand, Pareto distributions with infinite mean are not just theoretical constructs. They arise as useful models for representing certain rare but potentially extreme events, such as the losses from nuclear accidents \citep{hofert2012statistical, sornette2013exploring}, cyber and operational risks \citep{eling2019actual, eling2020capital, moscadelli2004modelling}, and fatalities from major earthquakes and pandemics \citep{clark2013note, cirillo2020tail}.\\
\\
Although the result of \citet{chen2025unexpected} is recent and has drawn interest, the phenomenon itself already appeared in earlier work. \citet{embrechts2002correlation}, for instance, proved a case of the stochastic dominance relation for $n = 2$ and a Pareto distribution with shape parameter $\alpha = 1/2$, and \citet{ibragimov2005new} established it for one-sided stable distributions with infinite mean. Besides this, asymptotic versions of the result have also been established. For example, for risks with regularly varying tails of index $\beta \in (0,1]$, \citet{albrecher2006tail} and \citet{embrechts2009multivariate} proved an inequality of the form
\begin{equation}\notag
\prob \left( \sum_{i=1}^n \theta_i X_i > x \right) \geq \prob(X > x)
\end{equation}
in the limit as $x \to \infty$.\\
\\
Returning to the work of \citet{chen2025unexpected}, the authors also extended their result in the same paper to weakly negatively associated and identically distributed (WNAID) super-Pareto risks. In this contribution and those that followed, the prefix \emph{super}- refers to a class of distributions obtained by applying a convex transformation to the base distribution, often with additional properties such as being increasing and anchored at zero. Subsequent work by \citet{chen2024stochastic} extended the stochastic dominance relation to negative lower orthant dependent (NLOD) risks within the so-called $\mathcal{H}$-family. Focusing on the iid case, \citet{muller2024some} and \citet{arab2024convex} established the result for super-Cauchy and InvSub risks, respectively; see also \citet{chen2025stochastic} for a closely related variant of the InvSub class (denoted $\mathcal H^*$ there).\\
\\
While each of those results assumed identically distributed risks, \citet{chen2025diversification} broadened the scope by relaxing this assumption and characterized diversification through majorization order. Given two positive exposure vectors $(\theta_1, \dots,  \theta_n)$ and $(\delta_1, \dots , \delta_n)$ with equal sum, the former is said to be smaller in majorization order if its components are less dispersed. Using this concept, they showed that the weighted sum of independent but not necessarily identically distributed Pareto risks with infinite mean becomes larger in the sense of first-order stochastic dominance as diversification increases according to the majorization order. \citet{chen2025stochastic} further extended this result to a larger distribution family, showing that this effect occurs beyond the Pareto case.\\
\\
\citet{chen2024stochastic} took a different approach to analyzing diversification with non-identically distributed risks by introducing a framework based on the generalized $r$-mean of the marginal distributions. They established a first-order stochastic dominance result comparing the generalized $r$-mean and the weighted average of NLOD risks with super-Fr{\'e}chet marginal distributions, which may differ, but all must share a common essential infimum equal to zero. A particularly relevant case occurs at $r = 1$, when the generalized $r$-mean of the marginal distributions becomes a mixture model. In that case, their result can be formulated as
\begin{equation}
\label{eq:mixture:rv}
 I_1 X_1 + \dots + I_n X_n \leq_{\text{st}} \theta_1 X_1 + \dots + \theta_n X_n,
\end{equation}
where $I_1, \dots, I_n$ are Bernoulli random variables such that exactly one of them equals $1$ and the rest are $0$, with $\prob(I_i = 1) = \theta_i$.

\subsection{On the mixture model}\label{sec:mixture}
The mixture model offers a meaningful generalization of the initial comparison between full exposure to a single risk and a diversified portfolio by preserving the idea of concentrating risk in a single position --- putting all eggs in one basket --- but doing so in a probabilistic manner. This ensures that the mixture remains a relevant benchmark for ``all-in" exposure even when the risks are not identically distributed. In the special case of identically distributed risks and the weights sum to $1$, the mixture follows the same distribution as any individual risk, recovering the original result in \eqref{eq:unexpected} as a particular instance of this broader framework.\\
\\
Moreover, the stochastic dominance relation~\eqref{eq:mixture:rv} directly implies
\begin{equation}\notag
\prob\Bigg(\sum_{i=1}^n \theta_i X_i > x \Bigg) \geq \sum_{i=1}^n \theta_i \,  \prob(X_i > x) \quad \text{for all } x,
\end{equation}
which relates the survival function of the weighted average --- typically lacking a closed-form expression --- to the corresponding marginals in a remarkably simple way. Whenever the dominance relation holds, this provides a tractable lower bound for exceedance probabilities.\\
\\
The comparison between a weighted average and the corresponding mixture model also arises in settings that involve randomization. One example appears in the actuarial literature on randomized reinsurance, which differs from standard contracts by introducing exogenous randomness into the coverage mechanism \citep{albrecher2019randomized, vincent2021structured, acciaio2025optimal}. Within this framework, the mixture can be viewed as a randomized scheme, where the reinsurer fully covers exactly one loss $X_i$, selected at random with probability $\theta_i$. The weighted average then corresponds to the deterministic counterpart: a set of quota-share treaties where the reinsurer pays a fixed fraction $\theta_i$ of each $X_i$. A second example comes from finance, where the weighted average represents a traditional diversified portfolio, while the mixture corresponds to time diversification: in each period, the investor allocates all capital to a single asset, selected at random according to the portfolio weights \citep{milevsky1998theoretical}.

\subsection{The content of this paper}
The present work builds upon the mixture framework. Our main result is the \textit{one-basket theorem}, which provides new sufficient conditions for the stochastic dominance relation~\eqref{eq:mixture:rv} to hold when the risks are independent but not necessarily identically distributed. In particular, for independent risks in the mixture setting, our contribution broadens the distributional scope of the result of \citet{chen2024stochastic} and removes the requirement that the risks share a common essential infimum.\\
\\
An important distinction from prior work is that the one-basket theorem does not require the stochastic dominance to hold for all admissible weight vectors. Earlier results identified classes of risks that satisfy relation~\eqref{eq:unexpected} or \eqref{eq:mixture:rv} uniformly across all such vectors. In contrast, our approach relaxes this requirement by providing conditions that can be checked for a specific weight vector, even if they are not met for others. This added flexibility allows us to establish stochastic dominance in cases beyond the reach of uniform results. A notable example is the discrete Pareto distribution with infinite mean: although the weighted average of such random variables may fail to dominate a single one under some weight allocations, we show that the equal-weighted average $\overline{X}_n = \frac{1}{n} \sum_{i=1}^n X_i$ satisfies $X \leq_{\text{st}} \overline{X}_n$. A related result holds for the average of St.\ Petersburg lotteries.\\
\\
Beyond the main theorem itself, we investigate the types of risks to which it applies. To that end, we introduce the concept of \textit{subscalability}, which provides a natural interpretation of the key inequalities appearing in the one-basket theorem. This notion serves as the foundation for defining $\theta$\textit{-subscalable} and \textit{completely subscalable} risks --- two related classes that satisfy the conditions of the theorem. We then analyze the properties of these classes, and relate them to existing families of risks for which dominance relations~\eqref{eq:unexpected} or \eqref{eq:mixture:rv} have been previously established.\\
\\
As a final insight, we place the one-basket theorem within a broader perspective. Given its sharp contrast with classical diversification principles, the result might initially appear anomalous. However, we show that it arises as the boundary case of a general phenomenon: diversification always increases the likelihood of exceeding small thresholds, and this local effect becomes global when the conditions of the theorem are satisfied, corresponding to first-order stochastic dominance.\\
\\
The paper is organized as follows. Section~\ref{sec:preliminaries} introduces the setting and notation. Section~\ref{section:convexorder} recalls the classical convex-order comparison between diversified and concentrated portfolios under finite expectations. Section~\ref{sec:mainresults} presents the main results, including the one-basket theorem. Section~\ref{sec:discussion} discusses connections to the literature and implications of the main result. Section~\ref{section:subscalability} introduces subscalability and studies classes of risks covered by the theorem. Section~\ref{section:beyond} provides the local-to-global perspective, and Section~\ref{conclusion} concludes. Proofs omitted from the main text, together with additional technical material, are collected in the appendix.

\section{Preliminaries}
\label{sec:preliminaries}
In what follows, we use the terms ``positive'' and ``increasing'' in the weak sense, that is, non-negative and non-decreasing, respectively, and all inequalities are understood in the non-strict sense. We write $\mathbb{N} = \lbrace 1, 2, 3, \dots \rbrace$ for the set of positive integers.\\
\\
We consider a probability space $(\Omega, \mathcal{F}, \mathbb{P})$ on which all relevant random variables are defined. The random variables of interest are $n \geq 2$ mutually independent positive random variables $X_1, \dots, X_n$, which we may refer to as risks.\\
\\
Throughout this paper, we work with the survival function of each risk $X_i$, defined as $\overline{F}_i(x) = \mathbb{P}(X_i > x)$. For completeness, we note that the cumulative distribution function is $F_i(x) = \prob(X_i \leq x)$, so that $\overline{F}_i(x) = 1 - F_i(x)$. If $X_i$ follows distribution $F_i$, we may write $X_i \sim F_i$ or equivalently $X_i \sim \overline{F}_i$. If two random variables $X$ and $Y$ have the same distribution, we write $X \sim Y$.\\
\\
A distribution that will appear repeatedly in the paper is the Pareto family. For $\alpha>0$ and $\rho>0$ we write $X\sim \mathrm{Pareto}(\alpha,\rho)$ if
\begin{equation}\notag
\Prob(X>x)=
\begin{cases}
1, & 0\le x<\rho,\\
\left(\dfrac{\rho}{x}\right)^{\alpha}, & x\ge \rho.
\end{cases}
\end{equation}
Note that $\mathbb{E}[X]=\infty$ if and only if $\alpha\le 1$.\\
\\
We define the weight vector as $\boldsymbol{\theta} := (\theta_1, \dots, \theta_n)$, which is assumed throughout to lie in the open probability simplex: 
\begin{equation}\notag
\Delta_n := \Bigg\lbrace (\theta_1, \dots, \theta_n) \in (0,1)^n : \sum_{i=1}^n \theta_i = 1 \Bigg\rbrace.
\end{equation}
The weighted average of the risks is
\begin{equation}\notag
\sum_{i=1}^n \theta_i X_i = \theta_1 X_1 + \dots + \theta_n X_n.
\end{equation}
It will be referred to as the \textit{diversified portfolio}.\\
\\
To define the corresponding mixture model, let $(I_1,\dots,I_n)\sim \mathrm{Categorical}(\boldsymbol{\theta})$, independent of $(X_1, \dots, X_n)$, so that exactly one of $I_1,\dots,I_n$ equals $1$ and the others equal $0$, with $\prob(I_i=1)=\theta_i$. The mixture is then 
\begin{equation}\notag
\sum_{i=1}^n I_i X_i = I_1 X_1 + \dots + I_n X_n,
\end{equation}
which selects exactly one of the $X_i$ at random according to the weights. In contrast to the diversified portfolio, the mixture concentrates all exposure on a single risk, and will therefore be referred to as the \textit{concentrated portfolio}.\\
\\
Although we assume that the weight vector $\boldsymbol{\theta}$ lies in $\Delta_n$, meaning that the total weight is exactly one, portfolios with total weight strictly less than one can still be represented. This can be done by setting one or more of the risks $X_i$ to be \textit{trivial} --- that is, almost surely zero.\\
\\
The survival function of the concentrated portfolio is given by 
\begin{equation}\notag
\prob \left( \sum_{i=1}^n I_i X_i > x \right) = \sum_{i=1}^n \theta_i \,  \prob(X_i > x) = \sum_{i=1}^n \theta_i \, \overline{F}_i(x).
\end{equation}
In the special case where the risks are identically distributed with common survival function $\overline{F}(x)$, this simplifies to 
\begin{equation}\notag
\prob \left( \sum_{i=1}^n I_i X_i > x \right) = \sum_{i=1}^n \theta_i \, \overline{F}(x) = \overline{F}(x),
\end{equation}
since the weights sum to one whenever $\boldsymbol{\theta} \in \Delta_n$. Thus, when the risks are identically distributed, the concentrated portfolio has the same distribution as any individual risk.\\
\\
Let $[n] := \lbrace 1, \dots, n \rbrace$. For any nonempty subset $\mu \subseteq [n]$, we define its \textit{subset weight} as
\begin{equation}\notag
    \theta_{\mu} := \sum_{i \in \mu} \theta_i.
\end{equation}
With this notation, a single subscript (e.g., $\theta_i$) refers to an individual weight, whereas a subscript corresponding to a nonempty subset (e.g., $\theta_{\mu}$) denotes the total weight over that subset.\\
\\
We now recall the definition of first-order stochastic dominance used throughout the paper. Given two random variables $X$ and $Y$, we say that $X$ is smaller than $Y$ in the sense of first-order stochastic dominance, written $X \le_{\mathrm{st}} Y$, if
\begin{equation}\label{eq:fosd-def}
\prob(X>x)\le \prob(Y>x)\qquad \text{for all }x\in\mathbb R.
\end{equation}
We write $X =_{\mathrm{st}} Y$ when \eqref{eq:fosd-def} holds with equality for all $x$, i.e.\ when $X \sim Y$. Finally, we write $X <_{\mathrm{st}} Y$ if $X \le_{\mathrm{st}} Y$ and $X\neq_{\mathrm{st}} Y$, that is, if \eqref{eq:fosd-def} is strict for at least one $x$. Since all random variables considered in this paper are non-negative, \eqref{eq:fosd-def} holds trivially on $(-\infty,0)$, where both survival functions equal $1$. Accordingly, whenever we establish first-order stochastic dominance, we will verify \eqref{eq:fosd-def} on $[0,\infty)$.\\
\\
We conclude this section by recalling the following two basic properties of first-order stochastic dominance, which will be used throughout.
\begin{lemma}\label{lemma:fosd}
(Closure under increasing transformations) If $X \leq_{\text{st}} Y$ and $f$ is any increasing function, then $f(X) \leq_{\text{st}} f(Y)$. In particular, scaling by a positive constant preserves $\leq_{\text{st}}$.\\
(Closure under convolution) Let $X_1, \dots, X_m$ and $Y_1, \dots, Y_m$ be two sets of independent random variables, and suppose that $X_i \leq_{\text{st}} Y_i$ for each $i = 1, \dots, m$. Then $X_1 + \dots + X_m \leq_{\text{st}} Y_1 + \dots + Y_m$.
\end{lemma}

\noindent For proofs, as well as a comprehensive treatment of stochastic orders, see the classic textbooks by \citet{muller2002comparison} and \citet{shaked2007stochastic}.

\section{The classical view on diversification}\label{section:convexorder}
In the introduction, we reviewed several results from the literature showing that, under specific conditions, diversification can in fact increase risk. Before turning to our main contribution, which further extends these findings, we first take a step back to reaffirm the classical view according to which diversification is beneficial. In particular, we formalize this intuition within our setting using the concept of convex order, which offers a natural framework to compare the diversified and concentrated portfolios.

\subsection{A reminder on convex order}
Convex order is a well-established method for comparing random variables and has numerous applications in risk analysis and decision theory. Formally, for two random variables $X$ and $Y$, we say that $X$ is smaller than $Y$ in convex order, denoted by $X \leq_{\text{cx}} Y$, if
\begin{equation}\notag
\mathbb{E}[\varphi(X)] \leq \mathbb{E}[\varphi(Y)]
\end{equation}
holds for all convex functions $\varphi$ for which the expectations exist.\\
\\
Convex order is usually considered in contexts where expectations are finite, in which case it can be interpreted as a way to compare the variability of random quantities that share the same expectation. This interpretation arises from two key observations. First, non-constant convex functions emphasize extreme values, so the inequality in the definition suggests that $Y$ tends to take on more extreme realizations than $X$. Second, since both functions $\varphi(x) = x$ and $\varphi(x) = -x$ are convex, it follows that if $X \leq_{\text{cx}} Y$ and the expectations exist, then $\mathbb{E}[X] = \mathbb{E}[Y]$. When, in addition, these expectations are finite, they provide a meaningful common reference point. Taken together, these observations confirm that convex order, in the finite-mean case, captures differences in variability around a shared expectation.\\
\\
This notion of risk contrasts with that captured by first-order stochastic dominance, which compares the relative location of distributions. Specifically, $X \leq_{\text{st}} Y$ means that $X$ is less likely than $Y$ to exceed any given threshold, reflecting a shift in probability mass toward smaller values. Roughly speaking, while first-order stochastic dominance captures a difference in size, convex order reflects (under finite expectations) differences in variability between random variables of the same size.\\
\\
Several important consequences follow from convex order, again assuming the relevant expectations are finite. For example, any risk-averse decision-maker will prefer the loss $X$ over $Y$ whenever $X \leq_{\text{cx}} Y$, since risk aversion corresponds to a utility function that is convex over the loss domain. Likewise, since the function $\varphi(x) = (x - c)^2$ is convex for any fixed $c \in \mathbb{R}$, the convex order relation implies that $X$ has a smaller variance than $Y$. This is closely related to the classical justification for diversification provided by the law of large numbers and Modern Portfolio Theory, as discussed in the introduction.\\
\\
Further applications of convex order, along with a detailed treatment, can be found in \citep{muller2002comparison} and \citep{shaked2007stochastic}.

\subsection{Convex order in our setting}
To show that a convex order relation holds in our setting, the key observation is that the diversified portfolio can be expressed as the conditional expectation of the concentrated one, given the random variables $\mathbf{X} := (X_1, \dots, X_n)$. That is, letting $P_C := \sum_{i=1}^n I_i X_i$ and $P_D := \sum_{i=1}^n \theta_i X_i$, we have
\begin{equation}\notag
\mathbb{E}[P_C | \mathbf{X}] = P_D.
\end{equation}
Applying Jensen's inequality to any convex function $\varphi$ for which expectations exist, we obtain
\begin{equation}\notag
\mathbb{E}[\varphi(P_D)] = \mathbb{E}[\varphi(\mathbb{E}[P_C | \mathbf{X}])] \leq \mathbb{E}[\mathbb{E}[\varphi(P_C) | \mathbf{X}]] = \mathbb{E}[\varphi(P_C)],
\end{equation}
which establishes the convex order relation.\\
\\
This argument yields the following lemma, a standard consequence of Jensen's inequality.
\begin{lemma}
Let $X_1, \dots, X_n$ be $n \geq 2$ independent positive random variables. Given a weight vector $\boldsymbol{\theta} \in \Delta_n$, let $(I_1, \dots, I_n) \sim \text{Categorical}(\boldsymbol{\theta})$, independent of $(X_1, \dots, X_n)$. Then
\begin{equation}\notag
\theta_1 X_1 + \cdots + \theta_n X_n \leq_{\text{cx}} I_1 X_1 + \cdots + I_n X_n.
\end{equation}
\end{lemma}
\noindent Hence, under finite expectations, the interpretation and consequences of convex order outlined in the previous subsection apply. In particular, the diversified portfolio exhibits less variability than the concentrated one while preserving the same mean, confirming the classical intuition that diversification reduces risk in well-behaved settings.\\
\\
By contrast, when expectations are infinite, they no longer provide meaningful information about the location of risks, and the interpretation above breaks down (see \citep{cote2025convex} for a detailed analysis of convex order under non-finite expectations). In these situations, counterintuitive phenomena can emerge, including cases where the diversified portfolio is actually larger than the concentrated one in first-order stochastic dominance, and thus presents more risk in that sense. The next section introduces the one-basket theorem, which provides sufficient conditions under which this surprising reversal occurs.

\section{Main results}
\label{sec:mainresults}
In this section, we present our core findings. We begin with the single-risk case, which paves the way for the multi-risk framework. We then establish a general inequality comparing the distribution of the diversified portfolio to that of the individual risks. This inequality serves as an intermediate step leading to the one-basket theorem, which provides sufficient conditions for the diversified portfolio to be larger than the concentrated one in the sense of first-order stochastic dominance. Connections, interpretation, and further implications are discussed separately in Section~\ref{sec:discussion}.

\subsection{The case of a single risk}

Consider the simple case where only one of the risks is non-trivial, and the others are almost surely zero. For ease of notation, we temporarily omit subscripts. The concentrated portfolio then reduces to $IX$, and the diversified portfolio simplifies to $\theta X$, where $I \sim \text{Bernoulli}(\theta)$ is independent of $X$, and $\theta \in (0,1)$.\\
\\
To determine when first-order stochastic dominance holds between these two portfolios, we compare their respective survival functions. The probability that $I X$ exceeds a threshold $x \geq 0$ is $\prob(I X > x) = \theta \, \overline{F}(x)$, while for $\theta X$, it is $\prob(\theta X > x) = \overline{F}(x/\theta)$. By the definition of first-order stochastic dominance, this leads to the following result:
\begin{lemma}\label{lemma:single}
Let $X$ be a positive random variable. Given $\theta \in (0,1)$, let $I \sim \text{Bernoulli}(\theta)$ be independent of $X$. Then $I X \leq_{\text{st}} \theta X$ if and only if the inequality $\theta \, \overline{F}(x) \leq \overline{F}(x/\theta)$ holds for all $x \geq 0$.
\end{lemma}

\noindent The inequality $\theta \overline{F}(x) \leq \overline{F}(x/\theta)$ thus fully determines when $I X$ is smaller than $\theta X$ in the sense of first-order stochastic dominance, and hints at a pattern that may extend to the multi-risk setting. In fact, it will be central to both the proof and the formulation of the next theorem, which requires keeping track of the points where this inequality holds, for each risk and across various subset weights. To that end, we define for each $X_i$ the region
\begin{equation}\notag
r_i(\theta) := \lbrace x \geq 0 : \theta \, \overline{F}_{i}(x) \leq \overline{F}_i(x/\theta)\rbrace, \qquad \theta \in (0,1).
\end{equation}
In Section \ref{section:subscalability}, we study two nested classes of risks defined by this inequality: those for which there exists $\theta\in(0,1)$ such that $\theta\,\overline F(x)\le \overline F(x/\theta)$ holds for all $x\ge 0$ ($\theta$-subscalable risks), and the smaller class for which this inequality holds for all $x\ge 0$ and all $\theta\in(0,1)$ (completely subscalable risks). Among the properties established in this section, Lemma~\ref{lemma:infinite} shows that whenever the inequality holds for some $\theta\in(0,1)$ for a non-trivial risk, the risk is heavy-tailed and has infinite mean.

\subsection{The general case}\label{subsection:generalcase}

When multiple risks are involved, extending the single-risk result requires a more refined approach. A careful partition of the sample space (see Lemma~\ref{lemma:partition} in the appendix) allows us to bound the survival function of the diversified portfolio from below over the region 
\begin{equation}\notag
\mathcal{R}(\boldsymbol{\theta}) := \bigcap_{i=1}^n \bigcap_{\lbrace i \rbrace \subseteq \mu \subsetneq [n]} r_i(\theta_{\mu}), \qquad \boldsymbol{\theta} \in \Delta_n,
\end{equation}
where $\theta_{\mu} = \sum_{i \in \mu} \theta_i$ is the subset weight of $\mu$, as previously defined.\\
\\
Over that region, the lower bound takes the form of a weighted average of the marginal survival functions, as established in the following theorem.

\begin{restatable}{theorem}{lowerbound}
\label{theorem:lowerbound}
Let $X_1, \dots, X_n$ be $n \geq 2$ independent positive random variables, with survival functions $\overline{F}_1, \dots, \overline{F}_n$. Given a weight vector $\boldsymbol{\theta} \in \Delta_n$, the inequality
\begin{equation}\notag
\prob\Bigg( \sum_{i = 1}^n \theta_i X_i > x \Bigg) \geq \sum_{i = 1}^n \theta_i \, \prob(X_i > x)
\end{equation}
holds for all $x \in \mathcal{R}(\boldsymbol{\theta})$.
\end{restatable}

\noindent The proof is given in Appendix~\ref{appendix:theorem}.\\
\\
That theorem has several key implications. Notably, as we show in Section~\ref{subsection:zeroT}, the region $\mathcal{R}(\boldsymbol{\theta})$ always includes a non-trivial interval $[0, t(\boldsymbol{\theta}))$, which ensures that the inequality holds for small values of $x$ in all cases. However, our main focus here is on the situations where the inequality extends to all positive values.\\
\\
Recall from Section~\ref{sec:preliminaries} that the survival function of the concentrated portfolio equals the weighted average of the marginal survival functions. Therefore, Theorem~\ref{theorem:lowerbound} implies that the inequality
\begin{equation}\notag
\prob\left( \sum_{i = 1}^n I_i X_i > x \right) \leq \prob\left( \sum_{i=1}^n \theta_i X_i > x \right)
\end{equation}
holds for all $x \geq 0$ whenever $\mathcal{R}(\boldsymbol{\theta}) = [0,\infty)$, so the diversified portfolio is larger than the concentrated portfolio in the sense of first-order stochastic dominance. By the definition of $\mathcal{R}(\boldsymbol{\theta})$, this occurs precisely when, for each $i \in [n]$ and every $\mu$ with $\lbrace i \rbrace \subseteq \mu \subsetneq [n]$, the inequality $\theta_{\mu} \, \overline{F}_i(x) \leq \overline{F}_i(x / \theta_{\mu})$ holds for all $x \geq 0$.\\
\\
This argument proves our main result:
\begin{theorem}[One-basket theorem]\label{theorem:onebasket}
Let $X_1, \dots, X_n$ be $n \geq 2$ independent positive random variables, with survival functions $\overline{F}_1, \dots, \overline{F}_n$. Given a weight vector $\boldsymbol{\theta} \in \Delta_n$, let $(I_1, \dots, I_n) \sim \text{Categorical}(\boldsymbol{\theta})$, independent of $(X_1, \dots, X_n)$. Suppose that for each $i \in [n]$ and every $\mu$ with $\lbrace i \rbrace \subseteq \mu \subsetneq [n]$, the condition
\begin{equation}\label{eq:onebasket:condition}
\theta_{\mu} \, \overline{F}_i(x) \leq \overline{F}_i(x / \theta_{\mu}) \qquad \text{for all } x \geq 0,
\end{equation}
is satisfied. Then
\begin{equation}\label{eq:onebasket:dominance}
I_1X_1 + \dots + I_nX_n \leq_{\text{st}} \theta_1X_1 + \dots + \theta_nX_n.
\end{equation}
\end{theorem}

\noindent The next corollary identifies the equality case in Theorem~\ref{theorem:onebasket}.

\begin{restatable}{corollary}{onebasketequality}
\label{cor:onebasket-equality}
Under the assumptions of Theorem~\ref{theorem:onebasket}, equality holds in~\eqref{eq:onebasket:dominance} if and only if $X_1=\cdots=X_n=0$ almost surely. Thus, if at least one of the random variables $X_1, \dots, X_n$ is non-trivial, then
\begin{equation}\notag
I_1 X_1 + \cdots + I_n X_n <_{\text{st}} \theta_1 X_1 + \cdots + \theta_n X_n.
\end{equation}
\end{restatable}
\noindent The proof is deferred to Appendix~\ref{appendix:cor:onebasket-equality}
\section{Discussion}\label{sec:discussion}
This section situates the one-basket theorem within the literature on diversification failures and discusses its implications. We proceed in three steps. First, we provide illustrations and connect our result to earlier contributions. Second, we relate our comparison between the diversified and concentrated portfolios to dependence-based notions of diversification, and explain why the ordering established by the one-basket theorem should not be interpreted as a dependence effect. Third, we discuss implications for randomized reinsurance.
\subsection{Illustrations and connections to prior work}
As already mentioned in the introduction, the mixture case (corresponding to the concentrated portfolio) was also studied by \citet{chen2024stochastic}, who established relation~\eqref{eq:onebasket:dominance} for NLOD (negative lower orthant dependent) super-Fr{\'e}chet risks with possibly different marginal distributions, under the assumption that all risks share a common essential infimum equal to zero. In contrast, our result assumes independence, but applies to a broader class of distributions (see Lemma~\ref{lemma:frechet}), and does not require a common essential infimum.\\
\\
The following example gives a concrete case where the result applies even though the risks have different essential infima.
\begin{example}[Non-identically distributed Pareto]\label{ex:noniid:pareto}
Let \(X_1,\dots,X_n\) be independent, where \(X_i\) is Pareto$(\alpha_i, \rho_i)$, with $\alpha_i \in (0,1]$ and $\rho_i > 0$ for all $i$. Since \(\alpha_i\le 1\), each \(X_i\) has infinite mean, and its essential infimum equals \(\rho_i\), which may depend on \(i\). Moreover, for each $i$ and for every $\theta \in (0,1)$, condition~\eqref{eq:onebasket:condition} holds. Hence it holds for all relevant subset weights \(\theta_\mu\in(0,1)\) whenever $\boldsymbol{\theta} \in \Delta_n$, so the one-basket theorem applies and yields \eqref{eq:onebasket:dominance}.
\end{example}
\vspace{.5em}
\noindent While this example highlights the flexibility of the theorem, our result also covers the traditional setting where the risks are identically distributed. In this scenario, the concentrated portfolio has the same distribution as any individual risk, and the comparison reduces to evaluating whether a single risk is stochastically dominated by a weighted average of its iid copies. As mentioned earlier, this case has been studied by \citet{ibragimov2005new}, \citet{chen2025unexpected}, \citet{chen2024stochastic}, \citet{muller2024some}, and \citet{arab2024convex}.\\
\\
A notable departure from prior work is that the one-basket theorem does not require the stochastic dominance relation to hold uniformly across all admissible weight vectors. Instead, it identifies a set of conditions that can be tested for any specific weight vector. If these conditions are satisfied, first-order stochastic dominance is guaranteed, regardless of whether they fail under other allocations. This relaxation allows us to establish dominance results for distributions beyond the reach of existing results.\\
\\
To illustrate this phenomenon, let $X_1,\ldots,X_n$ be iid copies of an integer-valued random variable $X$, with survival function
\begin{equation}\label{eq:dpar}
\overline{F}(x) = 
\begin{cases}
1, & x < 1,\\
\frac{1}{\lfloor x \rfloor + 1}, & x \geq 1,
\end{cases}
\end{equation}
where $\lfloor x \rfloor$ denotes the greatest integer less than or equal to $x$. This is a discretized Pareto random variable: $X$ is supported on $\lbrace 1, 2, \dots \rbrace$ and satisfies $X \sim \lfloor Y \rfloor$, where $Y$ is Pareto$(1,1)$. In the iid case, the concentrated portfolio satisfies $I_1 X_1 + \cdots + I_n X_n \sim X $, so relation~\eqref{eq:onebasket:dominance} reduces to \(X \le_{\mathrm{st}} \sum_{i=1}^n \theta_i X_i\). This dominance does not hold for all allocations: for instance, when \(n=2\) and $\theta_1 = 0.1$, we have
\(\prob(X>1.9) = \frac{1}{2} > \frac{61}{132} = \prob(\theta_1X_1+\theta_2X_2> 1.9 )\), so $X \not\le_{\mathrm{st}} \theta_1 X_1 + \theta_2 X_2$. However, for equal weights \(\theta_i=\frac{1}{n}\), the dominance does hold for all \(n\ge 2\), as stated next.
\begin{restatable}[Discrete Pareto]{proposition}{discpareto}
\label{prop:discpareto}
Let \(X\) be a random variable with survival function $\overline{F}$ defined in \eqref{eq:dpar}. Let further \(\overline X_n=\frac1n\sum_{i=1}^n X_i\) be the average of \(n\ge 2\) iid copies of \(X\). Then
\[
X \le_{\mathrm{st}} \overline X_n \qquad \text{for all } n\ge 2.
\]
\end{restatable}

\noindent The proof is given in Appendix~\ref{appendix:discrete}.
\begin{remark}\label{rem:discpareto-nonmonotone}
A natural question is whether Proposition~\ref{prop:discpareto} can be strengthened to
$\overline X_m \le_{\mathrm{st}} \overline X_n$ for all $1\le m\le n$. For iid $\mathrm{Pareto}(\alpha,\rho)$ risks with $\alpha\in(0,1]$, this ordering holds; see \citet[Corollary~1]{chen2025diversification}. For the discrete Pareto distribution~\eqref{eq:dpar}, however, the ordering fails for some pairs $m < n$. For instance, $\prob\!\left(\overline{X}_3 > \tfrac{5}{4}\right) = \tfrac{7}{8}
\;>\; \tfrac{41}{48}
= \prob\!\left(\overline{X}_4 > \tfrac{5}{4}\right)$, so $\overline X_3\not\le_{\mathrm{st}}\overline X_4$.
\end{remark}
\subsection{Diversification and dependence}\label{subsec:discussion-dependence}

The notion of diversification considered in this paper generalizes the classical comparison between a single risk and a diversified portfolio: the concentrated portfolio $\sum_{i=1}^n I_i X_i$ corresponds to taking full exposure to one component (selected at random), whereas the diversified portfolio $\sum_{i=1}^n \theta_i X_i$ takes simultaneous exposure to all components through a weighted average. Of course, diversification can be formalized in other ways. For example, a dependence-based view regards a portfolio as more diversified when its components exhibit less positive dependence (or more negative dependence),
e.g.\ when moving away from comonotonicity (perfect positive dependence) toward independence or a negative-dependence regime.\\
\\
While these two perspectives are conceptually different, they align in the identically distributed case. Assume that $X_1,\dots,X_n$ are identically distributed with $X_i \sim X$, and let $(X_1^*,\dots,X_n^*)$ be a comonotone coupling with the same marginals. Then $X_1^*=\cdots=X_n^*$ almost surely, so $\sum_{i=1}^n \theta_i X_i^* \sim X$ for any $\boldsymbol{\theta} \in \Delta_n$. Hence, 
\begin{equation}\notag
X\le_{\mathrm{st}}\sum_{i=1}^n\theta_iX_i
\qquad\Longleftrightarrow\qquad
\sum_{i=1}^n\theta_iX_i^* \le_{\mathrm{st}} \sum_{i=1}^n\theta_iX_i.
\end{equation}
Relation $X\le_{\mathrm{st}}\sum_{i=1}^n\theta_iX_i$ thus admits an equivalent dependence interpretation: replacing comonotonicity by the dependence structure of $(X_1,\dots,X_n)$ increases risk in
the sense of first-order stochastic dominance. In our setting this corresponds to independence, and related work shows that the same conclusion can hold under negative-dependence classes such as WNAID and NLOD \citep{chen2025unexpected,chen2024stochastic}. This runs counter to common intuition, since in finite-mean settings moving away from comonotonicity typically reduces risk (e.g.\ in convex order); see \citep{muller2002comparison, denuit2006actuarial, shaked2007stochastic}.\\
\\
With this perspective in mind, the one-basket theorem could be misread as pointing in the opposite direction. Indeed, the vector $(I_1X_1,\dots,I_nX_n)$ is mutually exclusive (at most one component is nonzero), hence strongly negatively dependent, whereas $(\theta_1X_1,\dots,\theta_nX_n)$ is independent. Thus, when
\[
I_1 X_1 + \cdots + I_n X_n \le_{\mathrm{st}} \theta_1 X_1 + \cdots + \theta_n X_n
\]
holds, one might be tempted to attribute the ordering to a hypothetical risk-reducing effect of negative dependence. This interpretation is incorrect, as we are not comparing dependence structures at fixed marginals: $I_iX_i$ and $\theta_iX_i$ generally have different distributions, so both the marginals and the dependence structure change across the two portfolios. 

\subsection{Randomized reinsurance}
Our setting reveals a noteworthy implication for randomized reinsurance, as studied by \citet{albrecher2019randomized, vincent2021structured, acciaio2025optimal}, and described in Section~\ref{sec:mixture}. For illustration, it is sufficient to consider the single-risk case. Let $X \sim \overline{F}$ represent an insurance loss. Under a classical quota-share treaty, the insurer cedes a fixed fraction $\theta X$ to the reinsurer and retains $(1 - \theta) X$. In contrast, consider a randomized reinsurance scheme in which the reinsurer covers the full loss $X$ with probability $\theta$, and pays nothing otherwise. The ceded and retained losses are then $I X$ and $(1 - I) X$, where $I \sim \text{Bernoulli}(\theta)$ is independent of $X$. In line with the result established in Lemma~\ref{lemma:single}, if the survival function $\overline{F}$ satisfies the inequalities
\begin{equation}\notag
\theta \, \overline{F}(x) \leq \overline{F}(x / \theta) \quad \text{and} \quad (1 - \theta) \, \overline{F}(x) \leq \overline{F}(x / (1 - \theta)) \quad \text{for all } x \geq 0,
\end{equation}
then both $I X \leq_{\text{st}} \theta X$ and $(1 - I) X \leq_{\text{st}} (1 - \theta) X$ hold. In that case, both the reinsurer and the insurer are better off under the randomized scheme in the sense of first-order stochastic dominance --- a perhaps surprising outcome, since adding exogenous randomness leads to better risk-sharing.\\
\\
Although we focused on the single-risk case for clarity, similar improvements may arise in multivariate settings, as shown by the one-basket theorem. Moreover, extensions to other randomized contract designs --- beyond the all-or-nothing coverage --- may also prove beneficial under suitable conditions, offering interesting directions for future research.

\section{Subscalability}\label{section:subscalability}

In the previous section, we showed that a first-order stochastic dominance relation holds between the diversified and concentrated portfolios when certain conditions on survival functions are met. Here, we explore the types of risks that satisfy these conditions by introducing the concept of subscalability, which underpins the definitions $\theta$-subscalable and completely subscalable risks --- two nested classes covered by the one-basket theorem. We then examine the properties of those two classes and relate them to risk families studied in earlier work.

\subsection{The concept of subscalability}\label{subsection:subscalability}

The one-basket theorem relies on a key condition: that the survival functions of the risks satisfy inequalities of the form $\theta \, \overline{F}(x) \leq \overline{F}(x/\theta)$, where $\theta \in (0,1)$. Intuitively, if the survival function of a risk $X$ satisfies this inequality at a given point $x$, it means that scaling down the risk by a factor $\theta$ reduces the probability of exceeding $x$ less than proportionally. That is,
\begin{equation}\notag
\theta \, \overline{F}(x) \leq \overline{F}(x/\theta) \quad \Longleftrightarrow \quad \theta \, \prob(X > x) \leq \prob(\theta X > x).
\end{equation}
At such points, we may say that the risk (or its survival function) resists scaling --- a property that we refer to as \emph{subscalability}.

\subsection{$\theta$-subscalable risks}
Building on this concept, we now define the class of \emph{$\theta$-subscalable} risks.
\begin{definition}\label{def:theta-subscalable}
For a given $\theta \in (0,1)$, a positive random variable $X \sim \overline{F}$ is said to be $\theta$-subscalable if the inequality $\theta \, \overline{F}(x) \leq \overline{F}(x/\theta)$ holds for all $x \geq 0$. We also refer to the corresponding distribution as $\theta$-subscalable.
\end{definition}

\noindent We begin with a basic implication of $\theta$-subscalability. Recall that a survival function $\overline F$ is
\textit{heavy-tailed} \citep{resnick2007heavy, embrechts2013modelling} if
\[
\lim_{x\to\infty} e^{t x}\,\overline F(x)=\infty \qquad \text{for all }t>0.
\]
\noindent Following \citet{embrechts2002correlation} and \citet{chen2025riskexchange}, we call a distribution \textit{extremely heavy-tailed} if it is heavy-tailed in the above sense and has infinite mean. The next lemma shows that, for non-trivial risks, $\theta$-subscalability is a sufficient condition for this
extreme tail behavior.
\begin{lemma}\label{lemma:infinite}
Let $\theta \in (0,1)$, and suppose $X \sim \overline{F}$ is a non-trivial $\theta$-subscalable random variable. Then $X$ is extremely heavy-tailed.
\end{lemma}
\begin{proof}
Since $X$ is positive and non-trivial, there exists $x_0>0$ such that $\overline F(x_0)>0$. Because $X$ is $\theta$-subscalable, Lemma~\ref{lemma:thetak} implies that $X$ is also $\theta^k$-subscalable for every $k\in\mathbb N$, that is,
\[
\theta^k\,\overline F(x)\le \overline F(x/\theta^k)\qquad \text{for all }x\ge 0.
\]
Evaluating this inequality at $x=x_0$ yields
\[
\overline F\!\left(\frac{x_0}{\theta^k}\right)\ge \theta^k\,\overline F(x_0), \qquad k\in\mathbb N.
\]
Now fix $x\ge x_0$ and choose $k\in\mathbb N$ such that
\[
\frac{x_0}{\theta^{k-1}}\le x<\frac{x_0}{\theta^{k}}.
\]
By monotonicity of $\overline F$,
\[
\overline F(x)\ge \overline F\!\left(\frac{x_0}{\theta^{k}}\right)\ge \theta^{k}\,\overline F(x_0).
\]
Moreover, $\frac{x_0}{\theta^{k-1}}\le x$ implies $\theta^{k-1}\ge x_0/x$, hence $\theta^{k}\ge \theta x_0/x$.
Therefore, letting $c:=\theta x_0\,\overline F(x_0)>0$, we obtain
\[
\overline F(x)\ge \frac{c}{x}\qquad \text{for all }x\ge x_0.
\]
It follows that for any $t>0$,
\[
e^{t x}\overline F(x)\ge c\,\frac{e^{t x}}{x}\xrightarrow[x\to\infty]{}\infty,
\]
so $X$ is heavy-tailed. Finally, since $X\ge 0$,
\[
\mathbb E[X]=\int_0^\infty \overline F(u)\,\diff u
\ge \int_{x_0}^\infty \frac{c}{u}\,\diff u=\infty,
\]
so $X$ has infinite mean. Hence $X$ is extremely heavy-tailed.
\end{proof}

\begin{remark}\label{rem:converse-false}
The converse of Lemma~\ref{lemma:infinite} is false: extreme heavy-tailedness alone does not guarantee $\theta$-subscalability. For instance, the survival function $\overline{F}(x)=\mathbf{1}_{\{x<e\}}+\mathbf{1}_{\{x\ge e\}}\frac{1}{x\ln x}$ defines a positive random variable with $\mathbb E[X]=\infty$ and a heavy tail in the above sense, yet it fails to be $\theta$-subscalable for any $\theta\in(0,1)$.
\end{remark}
\vspace{1mm}
\noindent Beyond these tail properties, the class of $\theta$-subscalable risks comprises a variety of distributions, including both discrete and continuous cases. Some distributions are $\theta$-subscalable for all $\theta\in(0,1)$, such as the Pareto family in Example~\ref{ex:noniid:pareto}. Distributions with this stronger property are studied further in Section~\ref{subsection:completely}. In contrast, other distributions are $\theta$-subscalable only for a restricted set of values of $\theta$. This phenomenon occurs for the St.\ Petersburg lottery (introduced below) and in discrete Pareto models --- see Lemma~\ref{lemma:stpetersburg} and Lemma~\ref{lemma:discreteParetoRegion} (appendix) for explicit characterizations of the admissible sets of $\theta$ in these two examples.\\
\\
A natural question then arises regarding the set of $\theta$ for which a given distribution is $\theta$-subscalable:
how small can it be? While this set can of course be empty for some distributions, the following lemma shows that
if it contains at least one value, then it must also contain an infinite sequence of smaller ones.

\begin{lemma}\label{lemma:thetak}
Let $\theta \in (0,1)$ and suppose that the random variable $X$ is $\theta$-subscalable. Then $X$ is also $\theta^{k}$-subscalable, for all $k \in \mathbb{N}$.
\end{lemma}
\begin{proof}
Since the inequality $\theta \, \overline{F}(x) \leq \overline{F}(x/\theta)$ holds for all $x \geq 0$, applying it at $x / \theta^j$, for each $j = 0, \dots, k-1$, gives
\begin{equation}\notag
\theta \, \overline{F}(x/\theta^{j}) \leq \overline{F}(x/\theta^{j+1}) \quad \text{for all } x \geq 0.
\end{equation}
Composing these inequalities yields
\begin{equation}\notag
\theta^k \, \overline{F}(x) \leq \overline{F}(x/\theta^{k}) \quad \text{for all } x \geq 0,
\end{equation}
which proves that $X$ is $\theta^k$-subscalable.
\end{proof}

\noindent Thus, subscalability at a single $\theta \in (0,1)$ entails subscalability at all smaller weights of the form $\theta^k$ with $k \in \mathbb{N}$, making the geometric sequence $\lbrace \theta^k : k \in \mathbb{N} \rbrace$ the minimal set for which subscalability must hold.\\
\\
Interestingly, for some distributions this minimal set is also exhaustive: subscalability holds along a geometric sequence, and for no other values of $\theta$. An example of this phenomenon is given by the St.\ Petersburg lottery. In this game, a fair coin is tossed until the first head appears; if the first head occurs on toss $m\in\mathbb N$, the payoff is $X=2^m$. Equivalently, $\prob(X=2^m)=2^{-m}$ for $m\in\mathbb N$, and the corresponding survival function is 
\begin{equation}\notag
\overline{F}(x) = 
\begin{cases}
1, & x < 2,\\
2^{- \lfloor \log_2 x \rfloor}, & x \geq 2.
\end{cases}
\end{equation}
As shown below, this distribution is $\theta$-subscalable if and only if $\theta=2^{-k}$ for some $k\in\mathbb N$. Hence, the admissible set of $\theta$ is the geometric sequence $\lbrace 2^{-k} : k \in \mathbb{N}\rbrace$, which, by Lemma~\ref{lemma:thetak}, is the minimal set generated by $\theta = \frac{1}{2}$.
\begin{lemma}\label{lemma:stpetersburg}
Let $X$ be a St.\ Petersburg lottery. Then $X$ is $\theta$-subscalable if and only if $\theta \in \{2^{-k}:k\in\mathbb N\}$.
\end{lemma}
\begin{proof}
\noindent\emph{Step 1: Sufficiency.}
Fix $k\in\mathbb N$ and set $\theta=2^{-k}$. We show that
\[
\theta\,\overline F(x)\le \overline F(x/\theta)\qquad \text{for all }x\ge 0.
\]
\noindent If $0\le x<2^{1-k}$, then $x < 2$ and $x/\theta=x2^k<2$, so $\overline F(x)=1$ and $\overline F(x/\theta)=1$. Therefore, 
\begin{equation}\notag
\theta\,\overline F(x)=\theta\le 1=\overline F(x/\theta).
\end{equation}
\noindent If $2^{1-k}\le x<2$, then $\overline F(x)=1$ and $2\le x/\theta=x2^k<2^{k+1}$, so $\bigl\lfloor \log_2(x/\theta)\bigr\rfloor\le k$, leading to
\[
\theta\,\overline F(x) = \theta = 2^{-k} \leq  2^{-\lfloor \log_2(x/\theta)\rfloor} = \overline F(x/\theta).
\]
\noindent Finally, if $x\ge 2$, choose $m\in\mathbb N$ such that $x\in[2^m,2^{m+1})$. Then
\[
\overline F(x)=2^{-m},
\qquad
x/\theta=x2^k\in[2^{m+k},2^{m+k+1}),
\]
hence
\[
\overline F(x/\theta)=2^{-(m+k)}=2^{-k}2^{-m}=\theta\,\overline F(x).
\]
\noindent Thus $X$ is $\theta$-subscalable for $\theta=2^{-k}$.

\medskip
\noindent\emph{Step 2: Necessity.}
Assume $\theta\in(0,1)$ is not of the form $2^{-k}$. Choose $k\in\mathbb N$ such that
\[
2^{-k}<\theta<2^{-k+1}.
\]
Set $x:=2^k\theta\in(1,2)$. Then $\overline F(x)=1$ while $x/\theta=2^k$, hence
$\overline F(x/\theta)=2^{-k}$. Since $\theta>2^{-k}$, we obtain
\[
\theta\,\overline F(x)=\theta>\,2^{-k}=\overline F(x/\theta),
\]
hence the inequality fails and $X$ cannot be $\theta$-subscalable.
\end{proof}
\noindent
Before moving on to the next section, we recall a stochastic dominance result for averages of St.\ Petersburg lotteries and include
a short proof using the one-basket theorem.

\begin{proposition}[\citet{chen2025diversification}]\label{proposition:stpetersburg}
Let $X$ be a St.\ Petersburg lottery, and let $\overline{X}_n = \frac{1}{n}\sum_{i=1}^n X_i$ be the average of $n\ge 2$ iid copies of $X$. If $k,\ell\in\mathbb N$ satisfy $\ell/k=2^a$ for some $a\in\mathbb N$, then
\[
\overline X_k \le_{\mathrm{st}} \overline X_\ell.
\]
\end{proposition}
\begin{proof}
By Lemma~\ref{lemma:stpetersburg}, $X$ is $\theta$-subscalable at $\theta=\tfrac12$, hence Theorem~\ref{theorem:onebasket} (applied with $n=2$ and equal weights) yields
\[
X \le_{\mathrm{st}} \overline X_2.
\]
Let $k\in\mathbb N$ and consider iid copies $X_1,\dots,X_{2k}$ of $X$. For each $i=1,\dots,k$, set
\[
Y_i:=\frac{X_i+X_{k+i}}{2}.
\]
Then $Y_1,\dots,Y_k$ are independent and each $Y_i$ has the same distribution as $\overline X_2$. Since $X_i \sim X$, it follows that $X_i \le_{\mathrm{st}} Y_i$ for every $i$. By closure of $\le_{\mathrm{st}}$ under convolution and positive scaling (Lemma~\ref{lemma:fosd}),
\[
\overline{X}_k=\frac{1}{k}\sum_{i=1}^k X_i
\le_{\mathrm{st}}
\frac{1}{k}\sum_{i=1}^k Y_i
=\frac{1}{k}\sum_{i=1}^k\frac{X_i+X_{k+i}}{2}
=\frac{1}{2k}\sum_{j=1}^{2k}X_j
=\overline{X}_{2k}.
\]
Iterating gives $\overline X_k \le_{\mathrm{st}} \overline X_{2^a k}$ for all $a\in\mathbb N$, i.e.\ whenever $\ell/k=2^a$.
\end{proof}
\begin{remark}
A natural question is whether $\overline X_k \le_{\mathrm{st}} \overline X_\ell$ can hold for pairs $(k,\ell)$ with $\ell/k$ not a power of two. \citet[Example~1]{chen2025diversification} give an explicit counterexample showing that $X \le_{\mathrm{st}} \overline X_3$ does not hold. A complete characterization of the weight vectors for which $X \le_{\mathrm{st}} \sum_{i=1}^n \theta_i X_i$ holds for iid St.\ Petersburg lotteries remains an open problem.
\end{remark}
\subsection{Completely subscalable risks}\label{subsection:completely}

A stronger condition than $\theta$-subscalability arises when the inequality $\theta \, \overline{F}(x) \leq \overline{F}(x/\theta)$ holds for all $x \geq 0$ and all $\theta \in (0,1)$. We define such risks as \emph{completely subscalable}.
\begin{definition}
A positive random variable $X$ is said to be completely subscalable if the inequality $\theta \overline{F}(x) \leq \overline{F}(x/\theta)$ holds for all $x \geq 0$ and all $\theta \in (0,1)$. We also refer to the corresponding distribution as completely subscalable.
\end{definition}
\noindent It follows directly from the definition that a completely subscalable risk is $\theta$-subscalable for all $\theta \in (0,1)$. By Lemma~\ref{lemma:infinite}, this implies that any non-trivial completely subscalable risk is extremely heavy-tailed.\\
\\
An equivalent characterization of complete subscalability is given by the following lemma.
\begin{lemma}\label{lemma:monotonicity}
A positive random variable $X \sim \overline{F}$ is completely subscalable if and only if the function $h(x) := x \, \overline{F}(x)$ is increasing on $(0,\infty)$.
\end{lemma}
\begin{proof}
Assume first that $X$ is completely subscalable. Then for any $0 < x < y$, set $\theta := x/y \in (0,1)$. The defining inequality gives $x \, \overline{F}(x) \leq y\, \overline{F}(y)$, which shows that $h$ is increasing on $(0,\infty)$. Conversely, suppose that $h$ is increasing on $(0,\infty)$. Fix $x > 0$ and let $y := x/\theta$. Then by monotonicity of $h$, we have $\theta \, \overline{F}(x) \leq \overline{F}(x/\theta)$. Since this inequality also holds at $x = 0$, we conclude that $X$ is completely subscalable.
\end{proof}

\noindent As a consequence, we obtain the following continuity property.
\begin{lemma}\label{lemma:continuous}
Let $X \sim \overline{F}$ be completely subscalable. Then $\overline{F}$ is continuous on $(0,\infty)$.
\end{lemma}
\begin{proof}
If $\overline{F}$ had a downward jump at any $x > 0$, then $h(x) = x \, \overline{F}(x)$ would jump down as well, contradicting the monotonicity property established in Lemma~\ref{lemma:monotonicity}. Hence $\overline{F}$ must be continuous on $(0,\infty)$.
\end{proof}

\noindent The next lemma establishes a useful closure property.
\begin{lemma}\label{lemma:closuresub}
Let $X \sim \overline{F}$ be completely subscalable and let $g : [0,\infty) \to [0,\infty)$ be an increasing convex function with $g(0) = 0$. Then the random variable $Y := g(X)$ is also completely subscalable.
\end{lemma}
\begin{proof}
Fix $x\ge 0$ and $\theta\in(0,1)$. Let $\overline{G}(u):=\mathbb{P}(Y>u)$ for $u\ge0$ and define the generalized inverse
\[
g^{\leftarrow}(u):=\sup\{t\ge 0:\ g(t)\le u\}, \qquad u\ge 0.
\]
If $g^{\leftarrow}(x)= \infty$ (which can only happen if $g\equiv0$, since $g$ is increasing convex with $g(0)=0$), then $\overline G(x)=\mathbb P(g(X)>x)=0$ and the desired inequality $\theta\,\overline G(x)\le \overline G(x/\theta)$ holds trivially. Thus assume $g^{\leftarrow}(x)<\infty$.\\
\\
Since $g$ is convex and finite-valued on $[0,\infty)$, it is continuous. Thus $\{g(X)>x\}=\{X>g^{\leftarrow}(x)\}$ and $\{g(\theta X)>x\}=\{X>g^{\leftarrow}(x)/\theta \}$, leading to
\[
\overline{G}(x)=\overline{F}\bigl(g^{\leftarrow}(x)\bigr), 
\qquad
\mathbb{P}\bigl(g(\theta X)>x\bigr)=\overline{F}\bigl(g^{\leftarrow}(x)/\theta\bigr).
\]
By convexity and $g(0)=0$, $g(\theta t)\le \theta g(t)$ for all $t\ge 0$, so
\[
\overline{G}(x/\theta)
=\mathbb{P}\bigl(\theta g(X)>x\bigr)
\ge \mathbb{P}\bigl(g(\theta X)>x\bigr)
= \overline{F}\bigl(g^{\leftarrow}(x)/\theta\bigr).
\]
Since $g^{\leftarrow}(x) \geq 0$, complete subscalability of $X$ yields
\[
\overline{F}\bigl(g^{\leftarrow}(x)/\theta\bigr)\ge \theta\,\overline{F}\bigl(g^{\leftarrow}(x)\bigr).
\]
Combining the above relations, we obtain $\theta\,\overline{G}(x)\le \overline{G}(x/\theta)$, so $Y$ is completely subscalable.
\end{proof}
\noindent We now relate the class of completely subscalable risks to the super-Fr{\'e}chet family introduced by \citet{chen2024stochastic}, under which their stochastic dominance result was derived. As mentioned earlier, the authors proved a broader stochastic dominance result involving the generalized $r$-mean and the weighted average of NLOD super-Fr{\'e}chet risks with possibly different distributions but common essential infimum of $0$. A special case of this result, corresponding to $r = 1$, yields the first-order stochastic dominance relation:
\begin{equation}\notag
I_1 X_1 + \cdots + I_n X_n \leq_{\text{st}} \theta_1 X_1 + \cdots + \theta_n X_n,
\end{equation}
as in the one-basket theorem.
\begin{definition}
A random variable $Y$ is said to be super-Fr{\'e}chet if it can be written as $Y = g^*(X)$, where $X \sim \text{Fr{\'e}chet}(1)$ and $g^* : [0,\infty) \to [0,\infty)$ is a strictly increasing convex function with $g^*(0) = 0$.
\end{definition}
\noindent For convenience, we denote the classes of super-Fr{\'e}chet and completely subscalable risks by $\mathcal{S}_{\mathcal{F}}$ and $\mathcal{CS}$, respectively.\\
\\
The next lemma shows that the Fr{\'e}chet$(1)$ distribution belongs to $\mathcal{CS}$.
\begin{lemma}
\textnormal{Fr{\'e}chet}$(1) \in \mathcal{CS}$.
\end{lemma}
\begin{proof}
Let $X \sim \text{Fr{\'e}chet}(1)$, whose survival function is $\overline{F}(x) = 1-e^{-1/x}$ for $x > 0$, with $\overline{F}(0) := 1$. Fix $\theta \in (0,1)$. Since the function $z \mapsto 1 - e^{-z}$ is concave on $(0,\infty)$, we have
\begin{equation}\notag
\theta (1 - e^{-1/x}) \leq 1 - e^{-\theta/x} \quad \Longrightarrow \quad  \theta \, \overline{F}(x) \leq \overline{F}(x/\theta),
\end{equation}
which verifies the inequality defining complete subscalability for all $x > 0$. Given that $\overline{F}(0) = 1$, the inequality also holds at $x = 0$. Hence the Fr{\'e}chet$(1)$ distribution is completely subscalable.
\end{proof}

\noindent This yields the following relationship between the two classes:
\begin{lemma}\label{lemma:frechet}
$\mathcal{S}_{\mathcal{F}} \subsetneq \mathcal{CS}$.
\end{lemma}
\begin{proof}
It follows from the two previous lemmas that $\mathcal{S}_{\mathcal{F}} \subseteq \mathcal{CS}$. This inclusion is strict: for instance, if $X \sim \text{Fr{\'e}chet}(1)$ and $g$ is convex, anchored at zero, and increasing but not strictly, then $Y = g(X)$ is completely subscalable, but not super-Fr{\'e}chet.
\end{proof}

\noindent Many classical heavy-tailed distributions belong to $\mathcal{S}_{\mathcal{F}}$. For example, the Burr, log-logistic, Lomax and paralogistic distributions with infinite mean are shown to be super-Fr{\'e}chet \citep[Example~4]{chen2024stochastic}. It follows from Lemma~\ref{lemma:frechet} that they are also completely subscalable.

\subsection{On the iid case}\label{subsection:iid}
As previously noted, when the risks are identically distributed, the concentrated portfolio has the same distribution as any of the marginal risks, and the dominance relation from the one-basket theorem simplifies to
\begin{equation}\label{eq:fosd:iid}
X \leq_{\text{st}} \theta_1 X_1 + \dots + \theta_n X_n.
\end{equation}
In the iid setting, if a random variable $X$ satisfies this relation, then so does any increasing convex transformation of $X$, say $f(X)$. That property follows from the preservation of first-order stochastic dominance under increasing transformations (Lemma~\ref{lemma:fosd}), together with Jensen's inequality:
\begin{equation}\notag
f(X) \leq_{\text{st}} f(\theta_1 X_1 + \dots + \theta_n X_n) \leq \theta_1 f(X_1) + \dots + \theta_n f(X_n).
\end{equation}
The implication above is a classical consequence of Jensen's inequality, which we state here for a fixed weight vector, in line with the setting of the one-basket theorem: 
\begin{lemma}\label{lemma:icx-transfer}
If relation~\eqref{eq:fosd:iid} holds for some weight vector $\boldsymbol{\theta} \in \Delta_n$, then we also have
\begin{equation}\notag
f(X) \leq_{\text{st}} \theta_1 f(X_1) + \dots + \theta_n f(X_n),
\end{equation}
for any increasing convex function $f$.
\end{lemma}
\noindent As already mentioned, the iid case received particular attention in earlier work, under the stronger requirement that the stochastic dominance relation holds for all admissible weight vectors. Following \citet[Definition~2.1]{muller2024some}, let $\mathcal D^{-}$ denote the class of survival functions $\overline{F}$ such that
\[
X \le_{\mathrm{st}} \sum_{i=1}^n \theta_i X_i
\qquad \text{for every } n \ge 2 \text{ and all } \boldsymbol{\theta}\in\Delta_n,
\]
whenever $X_1,\dots,X_n$ are iid copies of $X\sim \overline{F}$. With this notation, the goal in the above line of work is to identify explicit subclasses of $\mathcal D^{-}$.\\
\\
Early contributions include \citep{ibragimov2005new} for one-sided stable laws with infinite mean and \citep{chen2025unexpected} for super-Pareto risks, with further extensions to the $\mathcal H$-family in \citep{chen2024stochastic}. More recently, \citet{muller2024some} identified the super-Cauchy class, and \citet{arab2024convex} introduced the InvSub family, as additional explicit subclasses of $\mathcal D^{-}$. In the formulation of \citet{arab2024convex}, InvSub is defined under the condition \(F(0)=0\), so distributions with an atom at \(0\) are excluded. This restriction is removed in \citep{chen2025stochastic}, where the corresponding extension is denoted by $\mathcal H^*$. To maintain continuity with the InvSub terminology, we refer to this class as $\mathrm{InvSub}^*$.\\
\\
To date, the super-Cauchy and $\mathrm{InvSub}^*$ families are the largest explicit subclasses of $\mathcal D^{-}$ currently available. The two classes overlap but are not nested, as illustrated by examples in \citep[Section~4]{muller2024some} and \citep[Section~4.3]{arab2024convex}. A super-Cauchy risk is a random variable that can be written as $\varphi(Y)$, where $Y$ is standard Cauchy and $\varphi$ is convex. Importantly, $\varphi$ is not required to be increasing, anchored at zero, or even positive. As a result, the family includes random variables with support beyond the positive reals, such as the Cauchy distribution itself. On the other hand, the $\mathrm{InvSub}^*$ class can be expressed as the set of survival functions $\overline{F}$ with $\overline F(x)=1$ for $x<0$ such that
\[
\overline{F}(x/\theta) + \overline{F}(x/(1-\theta)) \geq \overline{F}(x)
\qquad \text{for all } x > 0 \text{ and all } \theta \in (0,1).
\]
With this formulation, it is straightforward to verify that $\mathcal{CS} \subseteq \mathrm{InvSub}^*$. Indeed, complete subscalability gives $\overline F(x/\theta)\ge \theta\,\overline F(x)$ and $\overline F(x/(1-\theta))\ge (1-\theta)\,\overline F(x)$ for all $x > 0$ and all $\theta \in (0,1)$, and summing yields the defining inequality above.  The inclusion is in fact strict, as shown in \citep[Example~3.4]{zeng2025further}. That paper also provides a systematic overview of the relationships between the principal distributional families studied in connection with diversification failures.

\section{The one-basket theorem as a boundary case}\label{section:beyond}

The one-basket theorem identifies sufficient conditions under which diversification increases risk in the sense of first-order stochastic dominance. This result contradicts classical intuition and may appear anomalous at first glance. However, in this section, we show that rather than being a pathological outlier, it arises as the boundary case of a general phenomenon: diversification always increases risk locally near the origin, and this effect becomes global when the conditions of the one-basket theorem are met.\\
\\
We proceed in three steps. First, we prove that every positive risk exhibits subscalability over a non-trivial interval near the origin. Next, we demonstrate that subscalability in this region gives rise to what can be seen as a local version of the one-basket theorem, valid in all cases near zero.  To measure how far this effect is guaranteed, we introduce the threshold function $t(\boldsymbol{\theta})$. Finally, we consider the case where $t(\boldsymbol{\theta})$ is infinite: in this setting, the effect extends globally, yielding the one-basket theorem.

\subsection{Subscalability near the origin}

We introduced the notion of subscalability in Section~\ref{subsection:subscalability}, to describe the situation in which scaling down a risk $X \sim \overline{F}$ by a factor $\theta \in (0,1)$ reduces the probability of exceeding a given threshold $x$ less than proportionally:
\begin{equation}\notag
\theta \, \prob(X > x) \leq \prob(\theta X > x) \quad \Longleftrightarrow \quad \theta \, \overline{F}(x) \leq \overline{F}(x/\theta).
\end{equation}
In the one-basket theorem, this inequality is required to hold globally --- that is, for all $x \geq 0$, across all risks and all relevant subset weights.\\
\\
Such a global requirement is strong and fails in many cases. However, we now show that the inequality holds locally near the origin, for any individual risk and for any $\theta \in (0,1)$.\\
\\
This simply follows from the right-continuity of survival functions: the inequality is trivially satisfied at $x = 0$, and by right-continuity of $\overline{F}$, it continues to hold on some interval $[0,\varepsilon)$ with $\varepsilon > 0$.\\
\\
To formalize it, we define for each risk $X_i$ and each $\theta \in (0,1)$ the threshold
\begin{equation}\notag
t_i(\theta) := \sup \lbrace t \geq 0 : [0,t) \subseteq r_i(\theta) \rbrace,
\end{equation}
where
\begin{equation}\notag
r_i(\theta) = \left\lbrace x \geq 0 : \theta \, \overline{F}_i(x) \leq \overline{F}_i(x / \theta) \right\rbrace.
\end{equation}
With this, the preceding argument leads to:
\begin{lemma}\label{lemma:ti}
Let $X_i \sim \overline{F}_i$ be a positive random variable. Then $t_i(\theta) > 0$ for any $\theta \in (0,1)$.
\end{lemma}

\subsection{A local version of the one-basket theorem}\label{subsection:zeroT}
While the one-basket theorem establishes first-order stochastic dominance under specific conditions, we now show that a local version of this result always holds: the diversified portfolio is more likely to exceed small thresholds than the concentrated one, regardless of whether the conditions of the theorem are satisfied.\\
\\
To see this, recall that by Theorem~\ref{theorem:lowerbound}, given any weight vector $\boldsymbol{\theta} \in \Delta_n$, the inequality 
\begin{equation}\label{ineq:localsum}
\prob\left( \sum_{i=1}^n \theta_i X_i > x \right) \geq \sum_{i=1}^n \theta_i \, \prob(X_i > x)
\end{equation}
holds for all $x$ in the region $\mathcal{R}(\boldsymbol{\theta})$, defined as
\begin{equation}\notag
\mathcal{R}(\boldsymbol{\theta}) = \bigcap_{i=1}^n \bigcap_{\lbrace i \rbrace \subseteq \mu \subsetneq [n]} r_i(\theta_{\mu}).
\end{equation}
\noindent By construction, each region $r_i(\theta_{\mu})$ contains the interval $[0,t_i(\theta_{\mu}))$. Therefore,
\begin{equation}\notag
\mathcal{R}(\boldsymbol{\theta}) \supseteq [0,t(\boldsymbol{\theta})), \quad \text{where} \quad 
t(\boldsymbol{\theta}) := \min_{i = 1}^n \min_{\lbrace i \rbrace \subseteq \mu \subsetneq [n]} t_i(\theta_{\mu}).
\end{equation}
Moreover, since each $t_i(\theta_{\mu})$ is strictly positive by Lemma~\ref{lemma:ti}, we have
\begin{equation}\notag
t(\boldsymbol{\theta}) > 0.
\end{equation}
Thus, inequality~\eqref{ineq:localsum} holds for all $x \in [0,t(\boldsymbol{\theta}))$, and since the right-hand side equals the survival function of the concentrated portfolio,
\begin{equation}\notag
\prob\left( \sum_{i=1}^n I_i X_i > x \right) = \sum_{i=1}^n \theta_i \, \prob(X_i > x),
\end{equation}
we obtain the following proposition:
\begin{proposition}\label{proposition:zerot}
Let $X_1, \dots, X_n$ be $n \geq 2$ independent positive random variables, with survival functions $\overline{F}_1, \dots, \overline{F}_n$. Given a weight vector $\boldsymbol{\theta} \in \Delta_n$, let $(I_1, \dots, I_n) \sim \text{Categorical}(\boldsymbol{\theta})$, independent of $(X_1, \dots, X_n)$. Then 
\begin{equation}\label{ineq:prop}
\prob\left( \sum_{i=1}^n I_i X_i > x \right) \leq \prob\left( \sum_{i=1}^n \theta_i X_i > x \right)
\end{equation}
holds for all $x \in [0,t(\boldsymbol{\theta}))$.
\end{proposition}
\noindent In line with the initial claim, this result establishes a local version of the one-basket theorem. Indeed, recall that the theorem asserts the stochastic dominance relation
\begin{equation}\notag
I_1 X_1 + \cdots + I_n X_n \leq_{\text{st}} \theta_1 X_1 + \cdots + \theta_n X_n,  
\end{equation}
which, by definition, is equivalent to the inequality~\eqref{ineq:prop} holding for all $x \geq 0$. In other words, when the theorem applies, the diversified portfolio is more likely than the concentrated one to exceed any given threshold. Proposition~\ref{proposition:zerot} shows that this behavior is always guaranteed locally: the inequality holds on a non-trivial interval $[0,t(\boldsymbol{\theta}))$, meaning that the diversified portfolio is more likely to exceed sufficiently small thresholds --- even when global stochastic dominance fails.
\begin{remark}\label{rem:local-strictness}
Proposition~\ref{proposition:zerot} is stated only with a non-strict inequality, because it is meant as the exact local counterpart of Theorem~\ref{theorem:onebasket}. The threshold \(t(\boldsymbol{\theta})\) naturally arises from the proof of the one-basket theorem and guarantees the inequality on an initial interval, but it is not designed to capture where strict inequality begins. A sharper local result, ensuring strict inequality at least at one point of an initial interval under mild assumptions, would require different arguments and a different threshold. We leave this for future work.
\end{remark}

\subsection{When the local effect becomes global}

As we just established, a local version of the one-basket theorem always holds near the origin. This is no coincidence: the two phenomena are directly connected. In fact, the one-basket theorem characterizes exactly the setting in which that local effect extends globally.\\
\\
This connection becomes clear upon examining Proposition~\ref{proposition:zerot} in the special case where $t(\boldsymbol{\theta}) = \infty$. Then, the proposition guarantees that inequality~\eqref{ineq:prop} holds for all $x \geq 0$, yielding first-order stochastic dominance:
\begin{equation}\notag
I_1 X_1 + \dots + I_n X_n \leq_{\text{st}} \theta_1 X_1 + \dots + \theta_n X_n.
\end{equation}
By definition, the threshold $t(\boldsymbol{\theta})$ is infinite if and only if the condition
\begin{equation}\notag
\theta_{\mu} \, \overline{F}_i(x) \leq \overline{F}_i(x/\theta_{\mu}) \qquad \text{for all } x \geq 0
\end{equation}
is satisfied for each $i \in [n]$ and for every $\mu$ with $\lbrace i \rbrace \subseteq \mu \subsetneq [n]$. These are exactly the conditions and the conclusion of the one-basket theorem.\\
\\
The theorem thus appears as a special instance of the proposition, confirming the connection between the local and global phenomena. Seen through this lens, the theorem emerges not as an anomaly, but as the boundary case of a broader effect: diversification always increases the likelihood of exceeding small thresholds, and while this is typically confined to a limited region near the origin, under specific conditions it extends globally, giving rise to first-order stochastic dominance.

\section{Conclusion}\label{conclusion}

We showed that, for extremely heavy-tailed risks, diversification can reverse in the strongest possible sense: the diversified portfolio may have larger exceedance probabilities at every threshold than the corresponding one-basket benchmark. Our main result --- the one-basket theorem --- provides simple, weight-specific sufficient conditions for this first-order stochastic dominance.\\
\\
To interpret these conditions, we introduced subscalability and the associated notions of $\theta$-subscalability and complete subscalability. These classes offer a transparent mechanism behind the dominance phenomenon and connect our setting to several distributional families previously identified in the literature. Importantly, the weight-specific nature of the theorem is particularly useful in discrete settings, where uniform results typically fail. It allowed us to obtain new dominance results for equal-weighted averages of infinite-mean discrete Pareto risks, and it yields an alternative proof of the stochastic dominance ordering for averages of St.\ Petersburg lotteries.\\
\\
Finally, we placed the one-basket theorem within a broader picture: diversification always increases the likelihood of exceeding sufficiently small thresholds, and the theorem pinpoints when this local effect propagates to all thresholds.

\appendix
\section*{Appendix}
\addcontentsline{toc}{section}{Appendix}

\section{Proof of Theorem~\ref{theorem:lowerbound}}\label{appendix:theorem}
Before turning to the main proof, we introduce some notation and establish a lemma that will be used throughout.\\
\\
Let $S \subseteq [n]$ be a non-empty index subset. For each $\mu \subseteq S$ define
\begin{equation}\notag
A_\mu(u) := \bigcap_{i \in \mu} \lbrace X_i > u \rbrace, \quad u \in \mathbb{R},
\end{equation}
which corresponds to the event that all variables with indices in $\mu$ exceed some threshold $u$. Note that $A_{\mu}(u)$ satisfies the monotonicity property
\begin{equation}\notag
A_{\mu}(u_1) \subseteq A_{\mu}(u_2), \qquad \text{for all } u_1 \geq u_2.
\end{equation}
\medskip
\noindent Next, let
\begin{equation}\notag
\mathbf{u} := (u_{\lambda} : \emptyset \subsetneq \lambda \subseteq S)
\end{equation}
be a vector of real thresholds indexed by the nonempty subsets of $S$, and define for each $\mu$ with $\emptyset \subsetneq \mu \subseteq S$ the event
\begin{equation}\notag
B_{\mu}(\mathbf{u}) := A_{\mu}(u_{\mu}) \cap \bigcap_{\mu \subsetneq \lambda \subseteq S} A_{\lambda}(u_{\lambda})^c.
\end{equation}
Here $E^c$ denotes the complement of an event $E$. Thus $B_{\mu}(\mathbf{u})$ occurs precisely when the variables indexed by $\mu$ exceed the threshold $u_{\mu}$, while no strict superset $\lambda \supsetneq \mu$ within $S$ exceeds its corresponding threshold $u_\lambda$. To complete the definition, we extend it to the case $\mu = \emptyset$,
\begin{equation}\notag
B_{\emptyset}(\mathbf{u}) := \bigcap_{\emptyset \subsetneq \lambda \subseteq S} A_{\lambda}(u_{\lambda})^c.
\end{equation}
\noindent The following lemma establishes that the events $B_{\mu}(\mathbf{u})$ partition the sample space under a monotonicity condition on the thresholds. This result is central to applying the law of total probability in the main argument.

\begin{lemma}\label{lemma:partition}
Let $S \subseteq [n]$, and suppose the thresholds satisfy the monotonicity condition
\[
u_{\mu} \geq u_{\lambda} \qquad \text{for all } \emptyset \subsetneq \mu \subsetneq \lambda \subseteq S.
\]
Then the collection $\lbrace B_{\mu}(\mathbf{u}) : \mu \subseteq S \rbrace$ forms a partition of the sample space $\Omega$.
\end{lemma}
\begin{proof}
We show that $\{B_\mu(\mathbf u):\mu\subseteq S \}$ is a partition of $\Omega$ by proving (i) exhaustiveness and (ii) pairwise disjointness.\\
\\
\emph{(i) Exhaustiveness.}
Fix $\omega\in\Omega$ and define
\[
M(\omega)
:=\{\mu:\ \emptyset\subsetneq\mu\subseteq S \ \text{and}\ \omega\in A_\mu(u_\mu)\},
\qquad
\mu^*(\omega):=\bigcup_{\mu\in M(\omega)}\mu.
\]

\smallskip
\noindent\emph{Case 1: $M(\omega)=\emptyset$.}
Then $\omega\notin A_\lambda(u_\lambda)$ for every $\emptyset\subsetneq\lambda\subseteq S$, hence $\omega\in B_\emptyset(\mathbf u)$.

\smallskip
\noindent\emph{Case 2: $M(\omega)\neq\emptyset$.}
We claim that $\omega\in B_{\mu^*(\omega)}(\mathbf u)$.

\smallskip
\noindent First, we show that $\omega\in A_{\mu^*(\omega)}(u_{\mu^*(\omega)})$.
Let $i\in\mu^*(\omega)$. By definition of $\mu^*(\omega)$, there exists $\mu_i\in M(\omega)$ such that $i\in\mu_i$.
Since $\omega\in A_{\mu_i}(u_{\mu_i})$, we have $\omega\in A_{\{i\}}(u_{\mu_i})$.
Moreover, because $\mu_i\subseteq \mu^*(\omega)$, the threshold assumption yields $u_{\mu_i}\ge u_{\mu^*(\omega)}$, and thus by the monotonicity of $A_{\{i\}}(u)$ in $u$, 
\[
\omega\in A_{\{i\}}(u_{\mu_i})\subseteq A_{\{i\}}(u_{\mu^*(\omega)}).
\]
Since this holds for all $i\in\mu^*(\omega)$, intersecting over $i$ gives
\[
\omega\in \bigcap_{i\in\mu^*(\omega)} A_{\{i\}}(u_{\mu^*(\omega)})
= A_{\mu^*(\omega)}(u_{\mu^*(\omega)}).
\]

\smallskip
\noindent Next, we show that $\omega\notin A_{\lambda}(u_{\lambda})$ for every $\lambda\supsetneq \mu^*(\omega)$.
Indeed, if $\omega\in A_{\lambda}(u_{\lambda})$ for some $\lambda\supsetneq \mu^*(\omega)$, then $\lambda\in M(\omega)$, hence
$\lambda\subseteq \bigcup_{\mu\in M(\omega)}\mu=\mu^*(\omega)$, a contradiction.

\smallskip
\noindent Thus $\omega\in A_{\mu^*(\omega)}(u_{\mu^*(\omega)})$ and $\omega\notin A_{\lambda}(u_{\lambda})$ for all $\lambda\supsetneq\mu^*(\omega)$,
which means $\omega\in B_{\mu^*(\omega)}(\mathbf u)$.

\smallskip
\noindent In both cases, $\omega$ belongs to some $B_\mu(\mathbf u)$, so
\[
\bigcup_{\mu\subseteq S}B_\mu(\mathbf u)=\Omega.
\]

\medskip
\noindent\emph{(ii) Pairwise disjointness.}
Fix $\mu_1\neq\mu_2$.

\smallskip
\noindent\emph{Case 1: one of $\mu_1,\mu_2$ is empty.}
Assume without loss of generality that $\mu_1=\emptyset$ and $\mu_2\neq\emptyset$.
Then $B_{\mu_2}(\mathbf u)\subseteq A_{\mu_2}(u_{\mu_2})$, while
$B_{\emptyset}(\mathbf u)\subseteq A_{\mu_2}(u_{\mu_2})^c$ by definition of $B_{\emptyset}(\mathbf u)$.
Hence $B_{\mu_1}(\mathbf u)\cap B_{\mu_2}(\mathbf u)=\emptyset$.

\smallskip
\noindent\emph{Case 2: $\emptyset\subsetneq\mu_1,\mu_2\subseteq S$.}
Suppose for contradiction that there exists $\omega\in B_{\mu_1}(\mathbf u)\cap B_{\mu_2}(\mathbf u)$.
Then $\omega\in A_{\mu_1}(u_{\mu_1})\cap A_{\mu_2}(u_{\mu_2})$. Set $\tau:=\mu_1\cup\mu_2$.
By the threshold assumption (since $\mu_1\subseteq\tau$ and $\mu_2\subseteq\tau$) we have $u_{\mu_1}\ge u_\tau$ and $u_{\mu_2}\ge u_\tau$,
hence $\omega\in A_\tau(u_\tau)$. Since $\mu_1\neq\mu_2$, we have $\tau\supsetneq\mu_1$ or $\tau\supsetneq\mu_2$, which contradicts the defining
requirement in $B_{\mu_1}(\mathbf u)$ or $B_{\mu_2}(\mathbf u)$ that no strict superset exceed its threshold.
Therefore $B_{\mu_1}(\mathbf u)\cap B_{\mu_2}(\mathbf u)=\emptyset$.

\smallskip
\noindent
In both cases the intersection is empty, hence the events $\{B_\mu(\mathbf u):\mu\subseteq S \}$ are pairwise disjoint.
\end{proof}

\noindent With this in place, we are ready to prove the theorem. For convenience, we restate it below.

\lowerbound*
\begin{proof}
The argument proceeds by applying the law of total probability, exploiting the partition structure, and using independence.\\
\\
Fix $x \in \mathcal{R}(\boldsymbol{\theta})$. Define
\begin{equation}\notag
u_{\mu}  := x/\theta_{\mu},
\end{equation}
and let
\begin{equation}\notag
\mathbf{u} = (u_{\mu} : \emptyset \subsetneq \mu \subseteq [n])
\end{equation}
be the resulting threshold vector, as in the definition of $B_{\mu}(\mathbf{u})$. Since $\boldsymbol{\theta} \in \Delta_n$, we have $0 < \theta_{\mu} \leq \theta_{\lambda}$ for all $\emptyset \subsetneq \mu \subsetneq \lambda \subseteq [n]$, and as $x \in \mathcal{R}(\boldsymbol{\theta})$ implies $x \geq 0$, the thresholds are well defined and satisfy the monotonicity condition
\begin{equation}\notag
u_{\mu} \geq u_{\lambda}.
\end{equation}
Thus, Lemma~\ref{lemma:partition} applies, and the events $\lbrace B_{\mu}(\mathbf{u}) : \mu \subseteq [n] \rbrace$ partition the sample space $\Omega$.\\
\\
Applying the law of total probability, we obtain
\begin{equation}\notag
\prob\Bigg(\sum_{i = 1}^n \theta_i X_i > x \Bigg) 
= \sum_{\mu \subseteq [n]} \prob\Bigg( \sum_{i = 1}^n \theta_i X_i > x, B_{\mu}(\mathbf{u})\Bigg),
\end{equation}
and dropping the term corresponding to $\mu = \emptyset$ thus yields the inequality
\begin{equation}\notag
\prob\Bigg(\sum_{i = 1}^n \theta_i X_i > x \Bigg) 
\geq \sum_{\emptyset \subsetneq \mu \subseteq [n]} \prob\Bigg( \sum_{i = 1}^n \theta_i X_i > x, B_{\mu}(\mathbf{u})\Bigg).
\end{equation}
Because each $\theta_i$ and each $X_i$ are positive, the weighted sum $\sum_{i \in \mu} \theta_i X_i$ can only increase when we add terms. Together with the given definitions, this results in the chain of inclusions 
\begin{equation}\notag
B_{\mu}(\mathbf{u}) 
\subseteq
A_{\mu}(u_{\mu})
\subseteq
\Bigg\lbrace \sum_{i \in \mu} \theta_i X_i > x \Bigg\rbrace
\subseteq
\Bigg\lbrace \sum_{i = 1}^n \theta_i X_i > x \Bigg\rbrace.
\end{equation}
Therefore $B_{\mu}(\mathbf{u}) \subseteq \lbrace \sum_{i = 1}^n \theta_i X_i > x \rbrace$, and we have
\begin{equation}\notag
\prob\Bigg( \sum_{i = 1}^n \theta_i X_i > x, B_{\mu}(\mathbf{u})\Bigg) = \prob(B_{\mu}(\mathbf{u})),
\end{equation}
which simplifies the inequality to
\begin{equation}\notag
\prob\Bigg(\sum_{i = 1}^n \theta_i X_i > x \Bigg) 
\geq \sum_{\emptyset \subsetneq \mu \subseteq [n]} \prob( B_{\mu}(\mathbf{u})).
\end{equation}
Noting $\sum_{i \in \mu} \frac{\theta_i}{\theta_{\mu}} = 1$ for each nonempty $\mu$, we may insert this identity into each summand, as
\begin{equation}\notag
\prob\Bigg(\sum_{i = 1}^n \theta_i X_i > x \Bigg) 
\geq
\sum_{\emptyset \subsetneq \mu \subseteq [n]} \sum_{i \in \mu} \frac{\theta_i}{\theta_{\mu}} \, \prob( B_{\mu}(\mathbf{u})).
\end{equation}
By changing the order of summation and appropriately redefining the indices, we obtain
\begin{equation}\notag
\prob\Bigg(\sum_{i = 1}^n \theta_i X_i > x \Bigg) 
\geq
\sum_{i=1}^n \sum_{\mu \subseteq [n] \setminus \lbrace i \rbrace} \frac{\theta_i}{\theta_{\mu \cup \lbrace i \rbrace}} \, \prob\!\big(B_{\mu \cup \lbrace i \rbrace}(\mathbf{u})\big).
\end{equation}
Now, fix $i\in[n]$ and $\mu \subseteq [n]\setminus\lbrace i\rbrace$. For $\emptyset\subsetneq\lambda\subseteq[n]\setminus\lbrace i\rbrace$, define the
thresholds $u^{(i)}_\lambda:=u_{\lambda\cup\lbrace i\rbrace}$ and set
\begin{equation}\notag
\mathbf u^{(i)}:=\big(u^{(i)}_\lambda:\ \emptyset\subsetneq\lambda\subseteq[n]\setminus\lbrace i\rbrace\big).
\end{equation}
Using the given definitions, we can then write
\begin{equation}\notag
\begin{aligned}
B_{\mu \cup \lbrace i \rbrace}(\mathbf{u})
&= A_{\mu\cup\lbrace i\rbrace}(u_{\mu\cup\lbrace i\rbrace})
   \cap \!\!\!\bigcap_{\mu\cup\lbrace i\rbrace \subsetneq \lambda \subseteq [n]} A_{\lambda}(u_{\lambda})^c \\
&= \lbrace X_i > u_{\mu\cup\lbrace i\rbrace}\rbrace \cap A_{\mu}(u_{\mu\cup\lbrace i\rbrace})
   \cap \!\!\!\bigcap_{\mu\subsetneq \lambda \subseteq [n] \setminus \lbrace i \rbrace} A_{\lambda}(u_{\lambda \cup \lbrace i \rbrace})^c \\
&= \lbrace X_i > x/\theta_{\mu \cup \lbrace i \rbrace} \rbrace \cap B_{\mu}(\mathbf{u}^{(i)}).
\end{aligned}
\end{equation}
Since $X_i$ and $X_j$ are independent whenever $i \neq j$ and $B_{\mu}(\mathbf{u}^{(i)})$ does not involve $X_i$ for all  $\mu \subseteq [n] \setminus \lbrace i \rbrace$, the events $\lbrace X_i > x / \theta_{\mu \cup \lbrace i \rbrace} \rbrace$ and $B_{\mu}(\mathbf{u}^{(i)})$ are independent. Consequently,
\begin{equation}\notag
\prob\!\big(B_{\mu \cup \lbrace i \rbrace}(\mathbf{u})\big)
= \overline{F}_i(x/\theta_{\mu \cup \lbrace i \rbrace}) \, \prob\!\big(B_{\mu}(\mathbf{u}^{(i)})\big).
\end{equation}
If $\mu \cup \lbrace i \rbrace \subsetneq [n]$, then $x \in \mathcal{R}(\boldsymbol{\theta})$ implies
\begin{equation}\notag
\overline{F}_i(x/\theta_{\mu \cup \lbrace i \rbrace}) \geq \theta_{\mu \cup \lbrace i \rbrace}  \overline{F}_i(x).
\end{equation}
If $\mu \cup \lbrace i \rbrace = [n]$, this inequality holds trivially (with equality), because $\theta_{[n]} = 1$. Hence, in all cases, 
\begin{equation}\notag
\prob(B_{\mu \cup \lbrace i \rbrace}(\mathbf{u})) 
\geq \theta_{\mu \cup \lbrace i \rbrace} \, \overline{F}_i(x) \, \prob(B_{\mu}(\mathbf{u}^{(i)})).
\end{equation}
As this lower bound on $\prob(B_{\mu \cup \lbrace i \rbrace}(\mathbf{u}))$ holds true for any $i \in [n]$ and $\mu \subseteq [n] \setminus \lbrace i \rbrace$, we can plug it into the preceding double-sum inequality, and we obtain
\begin{equation}\notag
\prob\Bigg(\sum_{i = 1}^n \theta_i X_i > x \Bigg) 
\geq 
\sum_{i=1}^n \theta_i \, \overline{F}_i(x) \sum_{\mu \subseteq [n] \setminus \lbrace i \rbrace}  \prob(B_{\mu}(\mathbf{u}^{(i)})).
\end{equation}
From the properties of the thresholds in $\mathbf{u}$ that we established earlier, it follows that the elements of $\mathbf{u}^{(i)}$ are also well-defined and satisfy the monotonicity condition. Lemma \ref{lemma:partition} therefore applies once more, meaning the events $\lbrace B_{\mu}(\mathbf{u}^{(i)}) : \mu \subseteq [n] \setminus \lbrace i \rbrace \rbrace$ partition the sample space $\Omega$, and we have
\begin{equation}\notag
\sum_{\mu \subseteq [n] \setminus \lbrace i \rbrace}  \prob(B_{\mu}(\mathbf{u}^{(i)})) = 1.
\end{equation}
This simplifies further the inequality to 
\begin{equation}\notag
\prob\Bigg(\sum_{i \in [n]} \theta_i X_i > x \Bigg) 
\geq \sum_{i = 1}^n \theta_i \overline{F}_i(x),
\end{equation}
which completes the proof.
\end{proof}

\section{Proof of Corollary~\ref{cor:onebasket-equality}}
\label{appendix:cor:onebasket-equality}
For convenience, we restate the corollary before proving it.

\onebasketequality*

\begin{proof}
For ease of notation, all equalities and inequalities between random variables in this proof are understood in the almost sure sense. Using the notation introduced in Section~\ref{section:convexorder}, we write
\[
P_C=\sum_{i=1}^n I_iX_i,
\qquad
P_D=\sum_{i=1}^n \theta_iX_i,
\qquad
\mathbf X=(X_1,\dots,X_n).
\]
If \(X_1=\cdots=X_n=0\), then clearly \(P_C=P_D=0\), so equality in~\eqref{eq:onebasket:dominance} holds. It remains to prove the converse.\\
\\
Assume therefore that \(P_C =_{\mathrm{st}} P_D\). Then \(P_C\sim P_D\). Consider the strictly convex function $x \mapsto e^{-x}$, and define
\[
Z:=\mathbb E[e^{-P_C}\mid \mathbf X]-e^{-\mathbb E[P_C\mid \mathbf X]},
\]
where $\mathbb{E}[P_C \mid \mathbf{X}] = P_D$. By Jensen's inequality,
\[
Z\ge 0.
\]
Moreover, since \(P_C\sim P_D\), we have $\mathbb E[e^{-P_C}]=\mathbb E[e^{-P_D}]$, and the tower property gives
\[
\mathbb E[Z]
=
\mathbb E\!\left[\mathbb E[e^{-P_C}\mid \mathbf X]\right]
-
\mathbb E\!\left[e^{-\mathbb E[P_C\mid \mathbf X]}\right]
=
\mathbb E[e^{-P_C}]-\mathbb E[e^{-P_D}]
=
0.
\]
Hence \(Z=0\), that is,
\[
e^{-\sum_{i=1}^n \theta_iX_i}
=
\sum_{i=1}^n \theta_i e^{-X_i}.
\]
Since \(x\mapsto e^{-x}\) is strictly convex and all weights \(\theta_i\) are strictly positive, this forces $X_1=\cdots=X_n$. Furthermore, since equal random variables that are independent must be constant, and since each $X_i$ is positive, there must exist a constant $c \geq 0$ such that
\begin{equation}\notag
X_1 = \cdots = X_n = c.
\end{equation}
It remains to show that \(c=0\). Suppose for contradiction that \(c>0\). Then for any \(i\in[n]\) and any \(x\in[\theta_i c,c)\),
\[
\theta_i\,\overline F_i(x)=\theta_i>0=\overline F_i(x/\theta_i),
\]
which contradicts condition~\eqref{eq:onebasket:condition} in Theorem~\ref{theorem:onebasket}, applied with \(\mu=\{i\}\). Hence \(c=0\), meaning that equality in~\eqref{eq:onebasket:dominance} holds if and only if
\[
X_1=\cdots=X_n=0 \qquad \text{almost surely}.
\]
The strict inequality statement then follows immediately from Theorem~\ref{theorem:onebasket} and the equality characterization above.
\end{proof}

\section{Proof of Proposition~\ref{prop:discpareto} (Discrete Pareto)}
\label{appendix:discrete}

This appendix proves Proposition~\ref{prop:discpareto}. The discrete Pareto distribution in the proposition
is supported on $\{1,2,\dots\}$ and has survival function
\begin{equation}\tag{\ref{eq:dpar}}
\overline F(x)= 
\begin{cases}
1, & x < 1, \\
\frac{1}{\lfloor x \rfloor + 1}, & x \ge 1.
\end{cases}
\end{equation}
It is convenient to work with the shifted variable $Y:=X-1$, which is supported on $\{0,1,2,\dots\}$, and has survival function
\begin{equation}\label{eq:disc:par}
\overline G(x)=
\begin{cases}
1, & x<0,\\
\frac{1}{\lfloor x\rfloor+2}, & x\ge 0.
\end{cases}
\end{equation}
We first identify the set of $\theta$ for which $Y$ is $\theta$-subscalable.

\begin{lemma}\label{lemma:discreteParetoRegion}
Let $Y$ be a random variable with survival function $\overline{G}$ defined in~\eqref{eq:disc:par}. Then $Y$ is $\theta$-subscalable if and only if $\theta \in \mathcal{A}$, where
\begin{equation}\notag
\mathcal{A} := \left(0, \tfrac{1}{2} \right] \cup \left\{ \tfrac{k+1}{2k+1} : k \in \mathbb{N} \right\}.
\end{equation}
\end{lemma}
\begin{proof}
First, observe that at $x = 0$, the inequality becomes $\theta \, \overline{G}(0) \leq \overline{G}(0)$, which trivially holds. Therefore, let $x > 0$ and define
\[
\ell := \lfloor x \rfloor + 1, \qquad m := \left\lfloor x/\theta \right \rfloor + 1.
\]
Since $\overline G(x)=1/(\lfloor x\rfloor+2)=1/(\ell+1)$ for $x\ge 0$, the inequality rewrites as
\[
\theta\frac{1}{\ell+1}\le \frac{1}{m+1}
\quad \Longleftrightarrow \quad
\theta (m+1)\le \ell+1.
\]
Fix $m\in\mathbb N$ and consider $x$ such that $m=\lfloor x/\theta\rfloor+1$, i.e. $x\in[(m-1)\theta,m\theta)$.
On this interval, $\ell=\lfloor x\rfloor+1$ is minimized at $x=(m-1)\theta$, giving $\ell=\lfloor(m-1)\theta\rfloor+1$.
Thus it suffices to check the inequality at that point, which yields
\begin{equation}\notag
\theta(m+1)\le \lfloor(m-1)\theta\rfloor+2
\quad \Longleftrightarrow \quad
\operatorname{frac}((m-1)\theta)\le 2(1-\theta),
\end{equation}
where $\operatorname{frac}(z)=z-\lfloor z\rfloor$.\\
\\
We now check this condition in three distinct cases.

\medskip
\noindent\emph{Case 1: $\theta\in(0,\tfrac12]$.}
Then $2(1-\theta)\ge 1$ and $\operatorname{frac}(\cdot)<1$, so the condition holds for all $m$.

\medskip
\noindent\emph{Case 2: $\theta \in (\tfrac{1}{2},1)$ irrational.}  
By Kronecker's theorem (see, e.g., Theorem 12.2.2 in \citep{miller2006}),
the set $\{ \operatorname{frac}((m - 1)\theta) : m \in \mathbb{N} \}$ is dense in $[0,1)$. Since $2(1 - \theta) < 1$, the condition fails for some $m$.

\medskip
\noindent\emph{Case 3: $\theta\in(\tfrac12,1)$ rational.}
Write $\theta=a/b$ with coprime $a<b$ and $2a>b$. Then $\operatorname{frac}((m-1)\theta)$ ranges over $\{0,1/b,\dots,(b-1)/b\}$,
so the condition for all $m$ requires $(b-1)/b\le 2(1-a/b)$, i.e. $2a\le b+1$. Together with $2a>b$ this forces $2a=b+1$,
hence $b=2k+1$ and $a=k+1$, i.e. $\theta=(k+1)/(2k+1)$.\\
\\
This proves the characterization of $\mathcal A$.
\end{proof}

\medskip
\noindent We now restate Proposition~\ref{prop:discpareto} and prove it.

\discpareto*

\begin{proof}
Let $X$ be as in the proposition and set $Y:=X-1$. Let $Y_1,\dots,Y_n$ be iid copies of $Y$, and write
\[
\overline Y_n:=\frac1n\sum_{i=1}^n Y_i.
\]
We first show that
\begin{equation}\label{eq:Y-dominance}
Y \le_{\mathrm{st}} \overline Y_n \qquad \text{for all } n\ge 2.
\end{equation}

\smallskip
\noindent\emph{Base cases.}
For $n=2$, the only relevant subset weight is $1/2\in\mathcal A$, so the one-basket theorem applies and yields
$Y\le_{\mathrm{st}}\overline Y_2$.
For $n=3$, the relevant subset weights are $1/3$ and $2/3$, and both belong to $\mathcal A$, so the one-basket theorem yields $Y\le_{\mathrm{st}}\overline Y_3$.

\smallskip
\noindent\emph{Induction step.}
Fix $n\ge 4$ and assume~\eqref{eq:Y-dominance} holds for all $j=2,\dots,n-1$.
Let $m:=\lfloor n/2\rfloor$ and define two independent block averages
\[
U_1:=\frac1m\sum_{i=1}^m Y_i,
\qquad
U_2:=\frac{1}{n-m}\sum_{i=m+1}^n Y_i.
\]
Then $U_1\sim \overline Y_m$ and $U_2\sim \overline Y_{n-m}$, and by the induction hypothesis
$Y\le_{\mathrm{st}} U_1$ and $Y\le_{\mathrm{st}} U_2$.
Let $Y_1',Y_2'$ be iid copies of $Y$, independent of $(U_1,U_2)$, and set
\[
\theta_1:=\frac{m}{n},\qquad \theta_2:=\frac{n-m}{n}.
\]
By Lemma~\ref{lemma:fosd} (closure under scaling and convolution),
\[
\theta_1 Y_1' + \theta_2 Y_2' \le_{\mathrm{st}} \theta_1 U_1 + \theta_2 U_2 = \overline Y_n.
\]
Moreover, $\theta_1,\theta_2\in\mathcal A$: if $n$ is even, $\theta_1=\theta_2=1/2$; if $n$ is odd, writing $n=2k+1$ gives
$\theta_1=k/(2k+1)$ and $\theta_2=(k+1)/(2k+1)$, and both lie in $\mathcal A$.
Therefore the one-basket theorem applies to the two-risk iid vector $(Y_1',Y_2')$ with weights $(\theta_1,\theta_2)$ and yields
\[
Y \le_{\mathrm{st}} \theta_1 Y_1' + \theta_2 Y_2'.
\]
Combining the last two inequalities gives $Y\le_{\mathrm{st}}\overline Y_n$, completing the induction and proving~\eqref{eq:Y-dominance} for all $n\ge 2$.

\medskip
\noindent\emph{Conclusion for the unshifted variable $X$.}
Since $X=Y+1$ and, for iid copies, $\overline X_n=\overline Y_n+1$, apply Lemma~\ref{lemma:icx-transfer} with the increasing convex
function $f(x)=x+1$ to~\eqref{eq:Y-dominance} to obtain
\[
X=f(Y)\le_{\mathrm{st}} \frac1n\sum_{i=1}^n f(Y_i)=\frac1n\sum_{i=1}^n (Y_i+1)=\overline Y_n+1=\overline X_n.
\]
This is exactly the claim of Proposition~\ref{prop:discpareto}.
\end{proof}

\section*{Acknowledgements}

The author thanks Hansjörg Albrecher for helpful comments on an earlier version of this manuscript.

\bibliographystyle{unsrtnat}
\bibliography{mybibliography}

\end{document}